\documentclass{CSML}

\def\dOi{12(4:1)2016}
\lmcsheading%
{\dOi}
{1--28}
{}
{}
{Apr.~20, 2016}
{Oct.~\phantom07, 2016}
{}

\ACMCCS{[{\bf Theory of computation}]: Logic; [{\bf Computing
    methodologies}]: Symbolic and algebraic manipulation---Symbolic
  and algebraic algorithms---Theorem proving algorithms; Artificial intelligence---Knowledge representation and reasoning}

\usepackage[utf8]{inputenc}
\usepackage[T1]{fontenc}
\usepackage[english]{babel}
\usepackage{microtype}
\usepackage{amsmath,amsfonts,amssymb}
\usepackage{stmaryrd}
\usepackage{todonotes}
\usepackage{xspace}
\usepackage{booktabs}
\usepackage{boxedminipage}

\usepackage[sort&compress,numbers]{natbib}
\usepackage{doi}
\usepackage{hyperref}

\setlist{leftmargin=2em,itemsep=.5ex}

\theoremstyle{definition}

\newcommand{\rulefont}[1]{{\fontfamily{cmss}\fontshape{sl}\selectfont #1}}

\newcommand{\ie}{i.e.\ }
\newcommand{\eg}{e.g.\ }
\newcommand{\wrt}{w.r.t.\ }
\newcommand{\cf}{c.f.\ }

\newcommand{\NP}{\textsc{NP}\xspace}
\newcommand{\ExpTime}{\textsc{ExpTime}\xspace}
\newcommand{\FLz}{\ensuremath{\mathcal{F\!L}_0}\xspace}
\newcommand{\EL}{\ensuremath{\mathcal{E\!L}}\xspace}
\newcommand{\NC}{\ensuremath{\mathsf{N}_{\mathsf{C}}}\xspace}
\newcommand{\NR}{\ensuremath{\mathsf{N}_{\mathsf{R}}}\xspace}
\newcommand{\Nv}{\ensuremath{\mathsf{N}_{\mathsf{v}}}\xspace}
\newcommand{\Nc}{\ensuremath{\mathsf{N}_{\mathsf{c}}}\xspace}
\newcommand{\Imc}{\ensuremath{\mathcal{I}}\xspace}

\newcommand{\nequiv}{\not\equiv}
\newcommand{\nsqsubseteq}{\not\sqsubseteq}
\newcommand{\At}{\ensuremath{\mathsf{At}}\xspace}
\newcommand{\NV}{\ensuremath{\mathsf{At_{nv}}}\xspace}
\newcommand{\rd}{\ensuremath{\mathsf{rd}}\xspace}
\newcommand{\Var}{\ensuremath{\mathsf{Var}}\xspace}
\newcommand{\sub}{\ensuremath{\mathfrak{s}}\xspace}
\newcommand{\Cl}{\ensuremath{\mathsf{Cl}}\xspace}

\newenvironment{algrule}[1][Rule]{%
\noindent\begin{minipage}{\textwidth}%
\rulefont{#1:}

\vspace{.2ex}%
\centering%
\begin{boxedminipage}{.98\textwidth}
\small%
}{%
\end{boxedminipage}%
\end{minipage}%
}
\newcommand{\condition}{\rulefont{Condition:}\xspace}
\newcommand{\action}{\vspace*{.2em}\rulefont{Action:}\xspace}
\newcommand{\rulespace}{\vspace{1ex}}

\title[Dismatching and Local Disunification in \EL]{Extending Unification in
\EL to Disunification: The Case of Dismatching and Local Disunification}
\thanks{Supported by DFG under grant BA 1122/14-2.}

\author[F.~Baader]{Franz Baader}
\address{Theoretical Computer Science, Technische Universität Dresden, Germany}
\email{\{Franz.Baader, Stefan.Borgwardt, Barbara.Morawska\}@tu-dresden.de}

\author[S.~Borgwardt]{Stefan Borgwardt}
\address{\vspace{-18 pt}}

\author[B.~Morawska]{Barbara Morawska}
\address{\vspace{-18 pt}}


\begin{document}

\begin{abstract}
  Unification in Description Logics has been introduced as a means to
  detect redundancies in ontologies.
  We try to extend the known decidability results for unification in the
  Description Logic \EL to disunification since negative constraints can be
  used to avoid unwanted unifiers.
  While decidability of the solvability of general \EL-disunification problems
  remains an open problem, we obtain \NP-completeness results for two
  interesting special cases: \emph{dismatching problems}, where one side of
  each negative constraint must be ground, and \emph{local} solvability of
  disunification problems, where we consider only solutions that are
  constructed from terms occurring in the input problem.
  More precisely, we first show that dismatching can be reduced to local 
  disunification,
  and then provide two complementary \NP-algorithms for finding local solutions
  of disunification problems.
\end{abstract}

\maketitle

\section{Introduction}

Description logics (DLs) \cite{BCNMP03} are a family of logic-based
knowledge representation formalisms, which can be used to represent
the conceptual knowledge of an application domain in a structured
and formally well-understood way. 
They are employed in various application
areas, but their most notable success so far is
the adoption of the DL-based language OWL~\cite{HoPH03}
as standard ontology language for the semantic web. 
DLs allow their users to define the important notions
(classes, relations) of the domain using concepts and roles;
to state constraints on the way these notions can be interpreted using 
terminological axioms; 
and to deduce consequences such as subsumption (subclass) relationships from 
the definitions and
constraints. The expressivity of a particular DL is determined by the
constructors available for
building concepts.
 
The DL \EL, which offers the concept
constructors conjunction ($\sqcap$), existential restriction ($\exists r.C$),
and the top concept ($\top$),
has drawn considerable attention in the last decade
since, on the one hand, important inference problems such as
the subsumption problem are polynomial in \EL, even with respect to expressive
terminological axioms~\cite{Bran04}. 
On the other hand, though quite inexpressive, \EL is used to define biomedical
ontologies, such as the large medical ontology SNOMED\,CT.\footnote{%
\url{http://www.ihtsdo.org/snomed-ct/}}
For these reasons, the most recent OWL version, OWL\,2,
contains the profile OWL\,2\,EL,\footnote{%
\url{http://www.w3.org/TR/owl2-profiles/}}
which is based on a maximally tractable extension of~\EL~\cite{BaBL08}.

Unification in Description Logics was introduced in~\cite{BaNa01}
as a novel inference service that can be used to detect redundancies in 
ontologies.
It is shown there that unification in the DL \FLz, which differs from
\EL in that existential restriction is replaced by value restriction
($\forall r.C$),
is \ExpTime-complete. The applicability of this result was not only hampered
by this high complexity, but also by the fact that \FLz is not
used in practice to formulate ontologies.

In contrast, as mentioned above, \EL is employed to build large biomedical 
ontologies for which detecting
redundancies is a useful inference service.
For example, assume that one developer of a medical ontology defines
the concept of a \emph{patient with severe head injury} as
\begin{eqnarray}
\label{term1}
\mathsf{Patient}\sqcap\exists\mathsf{finding}.(\mathsf{Head\_injury}\sqcap\exists\mathsf{severity}.\mathsf{Severe}),
\end{eqnarray}
whereas another one represents it as
\begin{eqnarray}
\label{term2}
\mathsf{Patient}\sqcap\exists\mathsf{finding}.(\mathsf{Severe\_finding}\sqcap\mathsf{Injury}\sqcap\exists\mathsf{finding\_site}.\mathsf{Head}).
\end{eqnarray}
Formally, these two concepts are not equivalent, but they are nevertheless 
meant to represent the same
concept. They can obviously be made equivalent by treating the concept names
$\mathsf{Head\_injury}$ and $\mathsf{Severe\_finding}$ as variables, and 
substituting the first one
by  $\mathsf{Injury}\sqcap\exists\mathsf{finding\_site}.\mathsf{Head}$
and the second one by $\exists\mathsf{severity}.\mathsf{Severe}$.
In this case, we say that the concepts are unifiable, and
call the substitution that makes them equivalent a \emph{unifier}.
In~\cite{BaMo-LMCS10}, we were
able to show that unification in \EL is of considerably lower
complexity than unification in \FLz: the decision problem for \EL 
is \NP-complete. 
The main idea underlying the proof of this result is to show that any solvable 
\EL-unification problem
has a local unifier, i.e., a unifier built from a polynomial number of 
so-called atoms determined by the unification problem.
However, the brute-force ``guess and then test'' \NP-algorithm obtained from
this result, which guesses a local substitution and then checks
(in polynomial time) whether it is a unifier, is not useful in practice.
We thus developed a goal-oriented
unification algorithm for~\EL, which is more efficient since nondeterministic
decisions are only made
if they are triggered by ``unsolved parts'' of the unification problem. 
Another option for obtaining a more efficient
unification algorithm is a translation to satisfiability in propositional 
logic (SAT): in~\cite{BaMo-LPAR10} it is shown
how a given \EL-unification problem~$\Gamma$ can be translated in
polynomial time into a propositional formula whose satisfying valuations
correspond to the local unifiers of $\Gamma$.

Intuitively, a unifier of two \EL concepts proposes definitions for the
concept names that are used as variables: in our example, we know that,
if we define $\mathsf{Head\_injury}$ as 
$\mathsf{Injury}\sqcap\exists\mathsf{finding\_site}.\mathsf{Head}$
and $\mathsf{Severe\_finding}$ as $\exists\mathsf{severity}.\mathsf{Severe}$,
then the two concepts (\ref{term1}) and (\ref{term2}) are equivalent w.r.t.\ 
these definitions.
Of course, this example was constructed such that the unifier (which is 
actually local) provides sensible definitions
for the concept names used as variables. In general, the existence of a 
unifier only says that
there is a structural similarity between the two concepts. The developer that 
uses unification as
a tool for finding redundancies in an ontology or between two different 
ontologies needs to inspect the
unifier(s) to see whether the definitions it suggests really make sense. For 
example, the substitution
that replaces $\mathsf{Head\_injury}$ by $\mathsf{Patient}\sqcap 
\mathsf{Injury}\sqcap\exists\mathsf{finding\_site}.\mathsf{Head}$
and $\mathsf{Severe\_finding}$ by $\mathsf{Patient}\sqcap 
\exists\mathsf{severity}.\mathsf{Severe}$ is also a local
unifier, which however does not make sense since findings (\ie
$\mathsf{Head\_Injury}$ or $\mathsf{Severe\_finding}$) cannot be patients.
Unfortunately, even small unification problems like the one in our example can
have too many local unifiers for manual inspection.
In~\cite{BaBM-AiML12} we propose to restrict the attention to so-called 
minimal unifiers, which form a subset of all local unifiers. 
In our example, the nonsensical unifier is
indeed not minimal. In general, however, the restriction to minimal unifiers 
may preclude interesting local unifiers. In addition,
as shown in~\cite{BaBM-AiML12}, computing minimal unifiers is actually harder 
than computing local unifiers (unless the polynomial
hierarchy collapses). 
In the present paper, we propose disunification as a more direct approach for 
avoiding local unifiers that do not make sense.
In addition to positive constraints (requiring equivalence or subsumption 
between concepts), a disunification problem may also
contain negative constraints (preventing equivalence or subsumption between 
concepts). In our example, the nonsensical unifier
can be avoided by adding the dissubsumption constraint 
\begin{eqnarray}
\label{dissubs:ex}
\mathsf{Head\_injury}\not\sqsubseteq^?\mathsf{Patient}
\end{eqnarray} 
to the equivalence constraint $(\ref{term1}) \equiv^? (\ref{term2})$.
We add a superscript $\cdot^?$ to the relation symbols (like~$\nsqsubseteq$
and~$\equiv$) to make clear that these are not axioms that are stated to hold,
but rather constraints that need to be solved by finding an appropriate
substitution.

Unification and disunification in DLs is actually a special case of unification
and disunification modulo equational theories (see~\cite{BaNa01}
and~\cite{BaMo-LMCS10} for the equational theories respectively corresponding 
to
\FLz and \EL).
Disunification modulo equational theories has, e.g., been investigated
in~\cite{BuerckertBuntine94,Comon91}. It is well-known
in unification theory that for effectively finitary equational theories, i.e., 
theories for which finite complete sets of unifiers 
can effectively be computed, disunification can be reduced to unification: to
decide whether a disunification problem has
a solution, one computes a finite complete set of unifiers of the equations and
then checks whether any of the unifiers in this set
also solves the disequations. Unfortunately, for \FLz and \EL, this approach 
is not feasible since the corresponding
equational theories have unification type zero~\cite{BaMo-LMCS10,BaNa01}, and
thus finite complete sets of unifiers need not even exist.
Nevertheless, it was shown in~\cite{BaOk12} that the approach used
in~\cite{BaNa01} to decide unification (reduction to
language equations, which are then solved using tree automata) can be adapted
such that it can also deal with disunification.
This yields the result that disunification in \FLz has the same 
complexity (\ExpTime-complete) as unification.

For \EL, going from unification to disunification appears to be more 
problematic. In fact, the main reason for unification to
be decidable and in \NP is locality: if the problem has a unifier then it has
a local unifier. We will show that
disunification in \EL is not local in this sense by providing an example of a 
disunification problem that has a solution, but
no local solution. Decidability and complexity of disunification in \EL
remains an open problem, but we provide partial solutions
that are of interest in practice.
On the one hand, we investigate
\emph{dismatching problems}, i.e., disunification problems where the negative
constraints are dissubsumptions $C\nsqsubseteq^? D$ for
which either $C$ or $D$ is ground (i.e., does not contain a variable).
Note that the dissubsumption (\ref{dissubs:ex}) from above actually satisfies
this restriction since $\mathsf{Patient}$ is
not a variable.  We prove that (general) solvability of
dismatching problems can be reduced to \emph{local disunification}, i.e., the
question whether a given \EL-disunification problem has a \emph{local} 
solution,
which shows that dismatching in \EL is \NP-complete.
On the other hand, we develop two specialized algorithms to solve local
disunification problems that extend the ones for
unification~\cite{BaMo-LPAR10,BaMo-LMCS10}: a goal-oriented algorithm that
reduces the amount of nondeterministic guesses necessary to find a local
solution, as well as a translation to SAT.
The reason we present two kinds of algorithms is that, in the case of 
unification, 
they have proved to complement each other well in first 
evaluations~\cite{BBMM-DL12}: the goal-oriented
algorithm needs less memory and finds minimal solutions faster, while the SAT
reduction generates larger data structures, but outperforms the goal-oriented
algorithm on unsolvable problems.

The remainder of this article is organized as follows.
Section~\ref{sec:subsumption} introduces syntax and semantics of \EL and
recalls some basic results about (dis)subsumption in \EL.
In Section~\ref{sec:disunification}, we introduce disunification and the special case of
unification, and recall known results about unification in~\EL and local solutions.
Section~\ref{sec:dismatching} contains our reduction from
dismatching to local disunification, while Sections~\ref{sec:rules}
and~\ref{sec:sat} describe the two algorithms for local disunification.
We discuss related work in Section~\ref{sec:related-work}, 
and summarize our results as well as sketch directions for future research in 
Section~\ref{sec:conclusions}.

This is an extended version of the conference paper~\cite{BaBM-RTA15}. In this
paper, we give full proofs of all our results, and add some results on how to
actually compute local solutions using the decision procedures presented in
Sections~\ref{sec:rules} and~\ref{sec:sat}.

\section{\texorpdfstring{Subsumption and dissubsumption in \EL}{Subsumption
and Dissubsumption in EL}}
\label{sec:subsumption}

The syntax of~\EL is defined based on two sets \NC and \NR of \emph{concept
names} and \emph{role names}, respectively. \emph{Concept terms} are
built from concept names using the constructors \emph{conjunction}
($C\sqcap D$), \emph{existential restriction} ($\exists r.C$ for $r\in\NR$),
and \emph{top}~($\top$).
An \emph{interpretation} $\Imc=(\Delta^\Imc,\cdot^\Imc)$ consists of a
non-empty domain~$\Delta^\Imc$ and an interpretation function that maps concept
names to subsets of~$\Delta^\Imc$ and role names to binary relations
over~$\Delta^\Imc$. This function is extended to concept terms as shown
in the semantics column of Table~\ref{tbl:el}.
\begin{table}[t]
\centering
\caption{Syntax and semantics of \EL}
\label{tbl:el}
\begin{tabular}{lcc}
  \toprule
  Name & Syntax & Semantics \\
  \midrule
  concept name & $A$ & $A^\Imc\subseteq\Delta^\Imc$
    \\ \addlinespace
  role name & $r$ & $r^\Imc\subseteq\Delta^\Imc\times\Delta^\Imc$
    \\ \addlinespace
  top & $\top$ & $\top^\Imc:=\Delta^\Imc$
    \\ \addlinespace
  conjunction & $C\sqcap D$ & $(C\sqcap D)^\Imc:=C^\Imc\cap D^\Imc$
    \\ \addlinespace
  existential restriction & $\exists r.C$
    & $(\exists r.C)^\Imc:=\{x\mid\exists y.(x,y)\in r^\Imc\land y\in C^\Imc\}$
    \\
  \bottomrule
\end{tabular}
\end{table}

A concept term~$C$ is \emph{subsumed} by a concept term~$D$
(written $C\sqsubseteq D$) if for every interpretation~\Imc it holds that
$C^\Imc\subseteq D^\Imc$.
We write a \emph{dissubsumption} $C\nsqsubseteq D$ to abbreviate the fact that
$C\sqsubseteq D$ does not hold.
The two concept terms $C$ and $D$ are \emph{equivalent} (written $C\equiv D$)
if $C\sqsubseteq D$ and $D\sqsubseteq C$,
\ie they are always interpreted as the same set.
The binary subsumption relation~$\sqsubseteq$ on concept terms is reflexive and
transitive, and $\equiv$ is an equivalence relation, which justifies the
notation.
Note that we use ``$=$'' to denote \emph{syntactic} equality between concept 
terms, whereas ``$\equiv$'' denotes semantic equivalence.

Since conjunction is interpreted as set intersection, we can treat $\sqcap$ as
a commutative and associative operator, and thus dispense with parentheses in
nested conjunctions.
An \emph{atom} is a concept name or an existential restriction.
Hence, every concept term~$C$ is a conjunction of atoms or~$\top$. We
call the atoms in this conjunction the \emph{top-level atoms} of~$C$.
Obviously, $C$ is equivalent to the conjunction of its top-level atoms, where
the empty conjunction corresponds to~$\top$.
An atom is \emph{flat} if it is a concept name or an existential restriction of
the form $\exists r.A$ with $A\in\NC$.

Subsumption in~\EL is decidable in polynomial time~\cite{BaKM-IJCAI99} and can
be checked by recursively comparing the top-level atoms of the two concept
terms.

\begin{lem}[\cite{BaMo-LMCS10}]
\label{lem:sub}
  For two atoms~$C,D$, we have $C\sqsubseteq D$ iff $C=D$ is a concept name or
  $C=\exists r.C'$, $D=\exists r.D'$, and $C'\sqsubseteq D'$.
  If $C,D$ are concept terms, then $C\sqsubseteq D$ iff for every
  top-level atom~$D'$ of~$D$ there is a top-level atom~$C'$ of~$C$ such that
  $C'\sqsubseteq D'$.
\qed
\end{lem}
We obtain the following contrapositive formulation characterizing
dissubsumption.

\begin{lem}
\label{lem:dissub}
  For two concept terms $C,D$, we have $C\nsqsubseteq D$ iff there is a
  top-level atom~$D'$ of~$D$ such that for all top-level atoms~$C'$ of~$C$ it
  holds that $C'\nsqsubseteq D'$.
\qed
\end{lem}
In particular, $C\nsqsubseteq D$ is characterized by the existence of a
top-level atom~$D'$ of~$D$ for which $C\nsqsubseteq D'$ holds.
By further analyzing the structure of atoms, we obtain the following.

\begin{lem}
\label{lem:dissub-atoms}
  Let $C,D$ be two atoms. Then we have $C\nsqsubseteq D$ iff either
  \begin{enumerate}
    \item $C$ or $D$ is a concept name and $C \not= D$; or
    \item $D = \exists r. D'$, $C = \exists s.C'$, and $r \not= s$; or
    \item $D = \exists r. D'$, $C = \exists r.C'$, and $C'\nsqsubseteq D'$.
      \qed
  \end{enumerate}
\end{lem}
%

\section{Disunification}
\label{sec:disunification}

As described in the introduction, we now partition the set~\NC into a set of
\emph{(concept) variables}~(\Nv) and a set of \emph{(concept) constants}~(\Nc).
A concept term is \emph{ground} if it does not contain any variables.
We define a quite general notion of disunification problems that is similar to
the equational formulae used in~\cite{Comon91}.

\begin{defi}
\label{def:disunification}
  A \emph{disunification problem} $\Gamma$ is a formula built from subsumptions
  of the form $C\sqsubseteq^?D$, where $C$ and $D$ are concept terms, using the
  logical connectives $\land$, $\lor$, and~$\lnot$.
  We use equations $C\equiv^?D$ to abbreviate
  $(C\sqsubseteq^?D)\land(D\sqsubseteq^?C)$, dissubsumptions
  $C\nsqsubseteq^?D$ for $\lnot(C\sqsubseteq^?D)$, and disequations
  $C\nequiv^?D$ instead of $(C\nsqsubseteq^?D)\lor(D\nsqsubseteq^?C)$.
  A \emph{basic disunification problem} is a conjunction of subsumptions and
  dissubsumptions.
  A \emph{dismatching problem} is a basic disunification problem in which all
  dissubsumptions $C\nsqsubseteq^?D$ are such that either $C$ or $D$ is ground.
  Finally, a \emph{unification problem} is a conjunction of subsumptions.
\end{defi}
To define the semantics of disunification problems, fix a \emph{finite
signature} $\Sigma\subseteq\NC\cup\NR$ and assume that all disunification
problems contain only concept terms constructed over the symbols in~$\Sigma$.
A \emph{substitution} $\sigma$ maps every variable in $\Sigma$ to a ground
concept term constructed over the symbols of~$\Sigma$. This mapping can be
extended to all concept terms (over~$\Sigma$) in the usual way.
A substitution $\sigma$ \emph{solves} a subsumption $C\sqsubseteq^?D$ if
$\sigma(C)\sqsubseteq\sigma(D)$; it \emph{solves} $\Gamma_1\land\Gamma_2$ if
it solves both~$\Gamma_1$ and~$\Gamma_2$; it solves $\Gamma_1\lor\Gamma_2$ if
it solves~$\Gamma_1$ or~$\Gamma_2$; and it solves~$\lnot\Gamma$ if it does
not solve~$\Gamma$.
A substitution that solves a given disunification problem is called a
\emph{solution} of this problem.
A disunification problem is \emph{solvable} if it has a solution.

By \emph{disunification} we refer to the decision problem of checking
whether a given disunification problem is solvable, and will similarly talk of
\emph{dismatching} and \emph{unification}.
In contrast to unification, in disunification it does make a difference 
whether or not
solutions may contain variables from $\Nv\cap\Sigma$ or additional symbols from
$(\NC\cup\NR)\setminus\Sigma$~\cite{BuerckertBuntine94}.
In the context of the application sketched in the introduction, restricting
solutions to ground terms over the signature of the ontology to be checked for
redundancy is appropriate: 
since a solution~$\sigma$ is supposed to provide definitions for the variables
in~$\Sigma$, it should not use the variables themselves to define them;
moreover, definitions that contain newly generated symbols would be meaningless
to the user.

\subsection{Reduction to basic disunification problems}
\label{sec:basic-disunification}

We will consider only basic disunification problems in the following.
The reason is that there is a straightforward \NP-reduction from solvability
of arbitrary disunification problems to solvability of basic disunification
problems.
In this reduction, we view all subsumptions occurring in the disunification
problem as propositional variables and guess a satisfying valuation of the
resulting propositional formula in nondeterministic polynomial time.
It then suffices to check solvability of the basic disunification problem
obtained as the conjunction of all subsumptions evaluated to true and the
negations of all subsumptions evaluated to false.
This reduction consists of polynomially many guesses followed by a polynomial
satisfaction check.
Hence, doing this before the \NP-algorithms for the problems considered in the
following sections leaves the overall complexity in \NP.
In fact, in contrast to the use of an \NP-oracle within an \NP-algorithm, all
the tests that are applied are deterministic polynomial time.
Overall, there are polynomially many guesses (in the reduction and the
\NP-algorithm) with deterministic polynomial tests at the end.

Hence, from now on we restrict our considerations to basic disunification
problems.
For simplicity, we will call them \emph{disunification problems} and consider
them to be \emph{sets} containing subsumptions and dissubsumptions.

\subsection{Reduction to flat disunification problems}
\label{sec:flat-disunification}

We further simplify our analysis by considering \emph{flat} disunification
problems, which means that they may only contain \emph{flat} dissubsumptions of
the form $C_1\sqcap\dots\sqcap C_n\nsqsubseteq^?D_1\sqcap\dots\sqcap D_m$ for
flat atoms $C_1,\dots,C_n,D_1,\dots,D_m$ with $m,n\ge 0$,%
\footnote{Recall that the empty conjunction is $\top$.}
and \emph{flat} subsumptions of the form
$C_1\sqcap\dots\sqcap C_n\sqsubseteq^? D_1$ for flat atoms $C_1,\dots,C_n,D_1$
with $n\ge 0$.
This restriction is without loss of generality:
to flatten concept terms, one can simply introduce new variables and equations
to abbreviate subterms~\cite{BaMo-LMCS10}. Moreover, a subsumption of the form
$C\sqsubseteq^?D_1\sqcap\dots\sqcap D_m$ is equivalent to the conjunction of
$C\sqsubseteq^?D_1$, \dots, $C\sqsubseteq^?D_m$.
Any solution of a disunification problem~$\Gamma$ can be extended to a solution
of the resulting flat disunification problem~$\Gamma'$, and conversely every
solution of~$\Gamma'$ also solves~$\Gamma$.

This flattening procedure also works for unification problems. However,
dismatching problems cannot without loss of generality be restricted to being
flat since the introduction of new variables to abbreviate subterms may destroy
the property that one side of each dissubsumption is ground (see also
Section~\ref{sec:dismatching}).

\subsection{Local disunification}
\label{sec:local-disunification}

For solving flat unification problems, it has been shown that it suffices to
consider so-called local solutions~\cite{BaMo-LMCS10}, which are restricted to
use only the atoms occurring in the input problem.
We define this notion here for disunification.

Let $\Gamma$ be a flat disunification problem. We denote by \At the set of
all (flat) atoms occurring as subterms in~$\Gamma$, by \Var the set of
variables occurring in~$\Gamma$, and by $\NV:=\At\setminus\Var$ the set of
\emph{non-variable atoms} of~$\Gamma$.
Let $S\colon\Var\to2^{\NV}$ be an \emph{assignment (for~$\Gamma$)}, \ie a
function that assigns to each variable $X\in\Var$ a set $S_X\subseteq\NV$ of
non-variable atoms.
The relation $>_S$ on \Var is defined as the transitive closure of
$\{(X,Y)\in\Var\times\Var \mid \text{$Y$ occurs in an atom of $S_X$}\}$.
If this defines a strict partial order, \ie $>_S$ is irreflexive, then $S$ is
called \emph{acyclic}. In this case, we can define the substitution
$\sigma_S$ inductively along $>_S$ as follows:
if $X$ is minimal w.r.t.~$>_S$, then all elements of $S_X$ are ground and we
simply take
\[ \sigma_S(X):=\bigsqcap_{D\in S_X}D; \]
otherwise, we assume that $\sigma_S(Y)$ is defined for all $Y\in\Var$ with
$X>_SY$, and set
\[ \sigma_S(X):=\bigsqcap_{D\in S_X}\sigma_S(D). \]
It is easy to see that the concept terms $\sigma_S(D)$ are ground and
constructed from the symbols of~$\Sigma$, and hence $\sigma_S$ is a valid
candidate for a solution of~$\Gamma$ according to
Definition~\ref{def:disunification}.

\begin{defi}
\label{def:local}
  Let $\Gamma$ be a flat disunification problem.
  A substitution $\sigma$ is called \emph{local} (w.r.t.~$\Gamma$) if there
  exists an acyclic assignment~$S$ for~$\Gamma$ such that $\sigma=\sigma_S$.
  The disunification problem $\Gamma$ is \emph{locally solvable} if it has a
  local solution, \ie a solution that is a local substitution.
  \emph{Local disunification} is the problem of checking flat disunification
  problems for local solvability.
\end{defi}
Note that assignments and local solutions are defined only for \emph{flat}
disunification problems, because both are based on the assumption that all
subterms occurring in the input problem are flat.
Although solvability of disunification problems is equivalent to solvability of
flat disunification problems, it is not straightforward to extend the notion of
local solutions to general disunification problems~$\Gamma$. In particular,
there may be several flat disunification problems that are equivalent
to~$\Gamma$ \wrt solvability, but they induce different sets of flat atoms, and
hence different kinds of local substitutions.

Obviously, local disunification is decidable in~\NP: We can guess an
assignment~$S$, and check it for acyclicity and whether the induced
substitution solves the disunification problem in polynomial time.
The corresponding complexity lower bound follows from \NP-hardness of (local)
solvability of unification problems in \EL~\cite{BaMo-LMCS10}.

\begin{fact}
\label{fact:np}
  Local disunification in \EL is \NP-complete.
\qed
\end{fact}
It has been shown that unification in~\EL is \emph{local} in the sense that the
equivalent flattened problem has a local solution iff the original problem is
solvable, and hence (general) solvability of unification problems in~\EL can be
decided in \NP~\cite{BaMo-LMCS10}.
The next example shows that disunification in~\EL is \emph{not local} in this
sense.

\begin{exa}
\label{exa:disunification-not-local}
  Consider the flat disunification problem
  \begin{align*}
    \Gamma := \{
      &X\sqsubseteq^?B,\ 
      A\sqcap B\sqcap C\sqsubseteq^?X,\ 
      \exists r.X\sqsubseteq^?Y,\ 
      \top\nsqsubseteq^?Y,\ 
      Y\nsqsubseteq^?\exists r.B \}
  \end{align*}
  with concept variables $X,Y$ and concept constants $A,B,C$. Then the
  substitution~$\sigma$ with $\sigma(X):=A\sqcap B\sqcap C$ and
  $\sigma(Y):=\exists r.(A\sqcap C)$ is a solution of~$\Gamma$.
  For~$\sigma$ to be local, the atom $\exists r.(A\sqcap C)$ would have to be
  of the form~$\sigma(D)$ for a non-variable atom~$D$ occurring in~$\Gamma$.
  But the only candidates for $D$ are $\exists r.X$ and $\exists r.B$, none of
  which satisfy $\exists r.(A\sqcap C)=\sigma(D)$.

  We show that $\Gamma$ cannot have another solution that is local. Assume to
  the contrary that $\Gamma$ has a local solution~$\gamma$. We know that
  $\gamma(Y)$ cannot be~$\top$ since~$\gamma$ must solve $\top\nsqsubseteq^?Y$.
  Furthermore, none of the constants $A,B,C$ can be a top-level atom
  of~$\gamma(Y)$ since this would contradict $\exists r.X\sqsubseteq^?Y$ (see
  Lemma~\ref{lem:sub}).
  That leaves only the non-variable atoms $\exists r.\gamma(X)$ and
  $\exists r.B$, which are, however, ruled out by $Y\nsqsubseteq^?\exists r.B$
  since both $\gamma(X)$ and $B$ are subsumed by~$B$ (see
  Lemma~\ref{lem:dissub-atoms}).
\end{exa}
The decidability and complexity of general disunification in \EL is still open.
In the following, we first consider the special case of solving dismatching
problems, for which we show a similar result as for unification: every
dismatching problem can be polynomially reduced to a flat problem that has a
local solution iff the original problem is solvable.
The main difference is that this reduction is nondeterministic.
In this way, we reduce dismatching to local disunification.
We then provide two different NP-algorithms for the latter problem by extending
the rule-based unification algorithm from~\cite{BaMo-LMCS10} and adapting the
SAT encoding of unification problems from~\cite{BaMo-LPAR10}. These algorithms
are more efficient than the brute-force ``guess and then test'' procedure on
which our argument for Fact~\ref{fact:np} was based.

\section{Reducing dismatching to local disunification}
\label{sec:dismatching}

Our investigation of dismatching is motivated in part by the work on
\emph{matching} in description logics, where similar restrictions are imposed
on unification problems~\cite{BKBM-JLC99,Kust-01_6,BaMo-KI14}.
In particular, the matching problems for~\EL investigated in~\cite{BaMo-KI14}
are similar to our dismatching problems in that there subsumptions are
restricted to ones where one side is ground.
Another motivation comes from our experience that dismatching problems already
suffice to formulate most of the negative constraints one may want to put on
unification problems, as described in the introduction.

As mentioned in Section~\ref{sec:disunification}, we cannot restrict our
attention to flat dismatching problems without loss of generality.
Instead, the nondeterministic algorithm we present in the following reduces any
dismatching problem~$\Gamma$ to a flat \emph{disunification} problem~$\Gamma'$
with the property that local solvability of~$\Gamma'$ is equivalent to the
solvability of~$\Gamma$.
Since the algorithm takes at most polynomial time in the size of~$\Gamma$, this
shows, together with Fact~\ref{fact:np}, that dismatching in~\EL is
\NP-complete.
For simplicity, we assume that the subsumptions and the non-ground sides of the
dissubsumptions have already been flattened using the approach mentioned in the
previous section. This retains the property that all dissubsumptions have one
ground side and does not affect the solvability of the problem.

Our procedure exhaustively applies a set of rules to the (dis)subsumptions in
a dismatching problem (see Figures~\ref{fig:dismatching-decomp}
and~\ref{fig:dismatching-flatten}).
Each rule consists of a \emph{condition} under which it is applicable to a
given subsumption or dissubsumption~\sub, and an \emph{action} that is executed
on~\sub. Actions usually include the removal of~\sub from the input problem,
and often new subsumptions or dissubsumptions are introduced to replace it.
Actions can \emph{fail}, which indicates that the current dismatching problem
has no solution.
In all rules, $C_1,\dots,C_n$ and $D_1,\dots,D_m$ denote atoms.
The rule \rulefont{Left Decomposition} includes the special case where the
left-hand side of~\sub is~$\top$, in which case \sub is simply removed from the
problem.
We use the rule \rulefont{Flattening Left-Ground Subsumptions} to eliminate the
non-flat, left-ground subsumptions that may be introduced by
\rulefont{Flattening Right-Ground Dissubsumptions}.

Note that at most one rule is applicable to any given (dis)subsumption.
The choice which (dis)subsumption to consider next is \emph{don't care}
nondeterministic, but the choices in the rules \rulefont{Right Decomposition}
and \rulefont{Solving Left-Ground Dissubsumptions} are \emph{don't know}
nondeterministic.
\begin{figure}[tb]
  \centering
  \setlength{\parskip}{0pt}
  \begin{algrule}[Right Decomposition]
    \condition This rule applies to
      $\sub=C_1\sqcap\dots\sqcap C_n\nsqsubseteq^?
        D_1 \sqcap\dots\sqcap D_m$
      if $m\neq 1$ and $C_1,\dots,C_n$, $D_1,\dots,D_m$ are atoms.

    \action If $m=0$, then \emph{fail}. Otherwise, choose an index
      $i\in\{1,\dots,m\}$ and replace \sub by
      $C_1\sqcap\dots\sqcap C_n\nsqsubseteq^?D_i$.
  \end{algrule}

  \rulespace
  \begin{algrule}[Left Decomposition]
    \condition This rule applies to
      $\sub=C_1\sqcap\dots\sqcap C_n\nsqsubseteq^?D$ if $n\neq 1$,
      $C_1,\dots,C_n$ are atoms, and $D$ is a non-variable atom.

    \action If $n=0$, then remove \sub from~$\Gamma$. Otherwise, replace \sub
      by $C_1\nsqsubseteq^?D$, \dots, $C_n\nsqsubseteq^?D$.
  \end{algrule}

  \rulespace
  \begin{algrule}[Atomic Decomposition]
    \condition This rule applies to $\sub=C\nsqsubseteq^?D$ if $C$ and $D$ are
      non-variable atoms.

    \action Apply the first case that matches~\sub:
      \begin{enumerate}[label=\alph*)]
        \item if $C$ and $D$ are ground and $C\sqsubseteq D$, then \emph{fail};
        \item if $C$ and $D$ are ground and $C\nsqsubseteq D$, then remove \sub
          from~$\Gamma$;
        \item if $C$ or $D$ is a constant, then remove~\sub
          from~$\Gamma$;
        \item if $C = \exists r. C'$ and $D = \exists s.D'$  with $r\neq s$,
          then remove \sub from $\Gamma$;
        \item if $C = \exists r. C'$ and $D = \exists r.D'$, then replace \sub
          by $C' \not\sqsubseteq^? D'$.
      \end{enumerate}
  \end{algrule}
  \vspace*{-2ex}
  \caption{Decomposition rules}
  \label{fig:dismatching-decomp}
\end{figure}
\begin{figure}[tb]
  \centering
  \begin{algrule}[Flattening Right-Ground Dissubsumptions]
    \condition This rule applies to
      $\sub=X\nsqsubseteq^?\exists r.D$ if $X$ is a variable and $D$ is ground
      and is not a concept name.

    \action Introduce a new variable~$X_D$ and replace \sub by
      $X\nsqsubseteq^?\exists r.X_D$ and $D\sqsubseteq^?X_D$.
  \end{algrule}

  \rulespace
  \begin{algrule}[Flattening Left-Ground Subsumptions]
    \condition This rule applies to
      $\sub=C_1\sqcap\dots\sqcap C_n\sqcap
        \exists r_1.D_1\sqcap\dots\sqcap \exists r_m.D_m\sqsubseteq^?X$
      if $m>0$, $X$ is a variable, $C_1,\dots,C_n$ are flat ground atoms, and
      $\exists r_1.D_1,\dots,\exists r_m.D_m$ are non-flat ground atoms.

    \action Introduce new variables $X_{D_1},\dots,X_{D_m}$ and replace \sub by
      $D_1\sqsubseteq^?X_{D_1}$, \dots, $D_m\sqsubseteq^?X_{D_m}$, and
      $C_1\sqcap\dots\sqcap C_n\sqcap
        \exists r_1.X_{D_1}\sqcap\dots\sqcap\exists r_m.X_{D_m}\sqsubseteq^?X$.
  \end{algrule}

  \rulespace
  \begin{algrule}[Solving Left-Ground Dissubsumptions]
    \condition This rule applies to
      $\sub=C_1\sqcap\dots\sqcap C_n\nsqsubseteq^?X$ if $X$ is a variable and
      $C_1,\dots,C_n$ are ground atoms.

    \action Choose one of the following options:
      \begin{itemize}
        \item Choose a concept constant $A\in\Sigma$ and replace \sub by
          $X\sqsubseteq^?A$. If $C_1\sqcap\dots\sqcap C_n\sqsubseteq A$, then
          \emph{fail}.
        \item Choose a role $r\in\Sigma$, introduce a
          new variable~$Z$, replace \sub by $X\sqsubseteq^?\exists r.Z$,
          $C_1\nsqsubseteq^?\exists r.Z$, \dots,
          $C_n\nsqsubseteq^?\exists r.Z$, and immediately apply
          \rulefont{Atomic Decomposition} to each of these dissubsumptions.
      \end{itemize}
  \end{algrule}
  \vspace*{-2ex}
  \caption{Flattening and solving rules}
  \label{fig:dismatching-flatten}
\end{figure}

\begin{algo}
\label{alg:dismatching-reduction}
  Let $\Gamma_0$ be a dismatching problem. We initialize $\Gamma:=\Gamma_0$.
  While any of the rules of Figures~\ref{fig:dismatching-decomp}
  and~\ref{fig:dismatching-flatten} is applicable to any element of~$\Gamma$,
  choose one such element and apply the corresponding rule. If any rule
  application fails, return ``failure''.
\end{algo}
Note that each rule application takes only polynomial time in the size of the
chosen (dis)subsumption. In particular, subsumptions between ground atoms can
be checked in polynomial time~\cite{BaKM-IJCAI99}.

\begin{lem}
\label{lem:dismatching-termination}
  Every run of Algorithm~\ref{alg:dismatching-reduction} terminates in time
  polynomial in the size of~$\Gamma_0$.
\end{lem}
\proof
  Let $\Gamma_0$, \dots, $\Gamma_k$ be the sequence of disunification problems
  created during a run of the algorithm, \ie
  \begin{itemize}
    \item $\Gamma_0$ is the input dismatching problem;
    \item for all $j$, $0\le j\le k-1$, $\Gamma_{j+1}$ is the result of
      successfully applying one rule to a (dis)subsumption in~$\Gamma_j$; and
    \item either no rule is applicable to any element of~$\Gamma_k$, or a rule
      application to a (dis)subsumption in~$\Gamma_k$ failed.
  \end{itemize}
  We prove that $k$ is polynomial in the size of~$\Gamma_0$ by measuring the
  size of subsumptions and dissubsumptions via the function~$c$ that is defined
  as follows:
  \[ c(C\nsqsubseteq^?D):=c(C\sqsubseteq^?D):=|C|\cdot|D|, \]
  where $|C|$ is the size of the concept term~$C$; the latter is measured in
  the number of symbols it takes to write down~$C$, where we count each concept
  name as one symbol, and ``$\exists r.$'' is also one symbol.
  Note that we always have $|C|\ge 1$ since $C$ must contain at least one
  concept name or~$\top$, and thus also $c(\sub)\ge 1$ for any
  (dis)subsumption~\sub.
  We now define the size $c(\Gamma)$ of a disunification problem~$\Gamma$ as
  the sum of the sizes $c(\sub)$ for all $\sub\in\Gamma$ to which a rule is
  applicable.

  Since $c(\Gamma_0)$ is obviously polynomial in the size of~$\Gamma_0$, it
  now suffices to show that $c(\Gamma_j)>c(\Gamma_{j+1})$ holds for all $j$,
  $0\le j\le k-1$.
  To show this, we consider the rule that was applied to $\sub\in\Gamma_j$ in
  order to obtain~$\Gamma_{j+1}$:
  \begin{itemize}
    \item \rulefont{Right Decomposition:} Then
      $\sub=C_1\sqcap\dots\sqcap C_n\nsqsubseteq^?D_1\sqcap\dots D_m$ and we
      must have $m>1$ since we assumed that the rule application was
      successful.\par
      Thus, we get
      \[|C_1\sqcap\dots\sqcap C_n|\cdot|D_1\sqcap\dots\sqcap D_m|
        > |C_1\sqcap\dots\sqcap C_n|\cdot|D_i|\] for every choice of
      $i\in\{1,\dots,m\}$, and hence $c(\Gamma_j)>c(\Gamma_{j+1})$.
    \item \rulefont{Left Decomposition:} Then
      $\sub=C_1\sqcap\dots\sqcap C_n\nsqsubseteq^?D$ and, if $n=0$, we
      therefore have $c(\Gamma_j)=c(\Gamma_{j+1})+c(\sub)
        \ge c(\Gamma_{j+1})+1>c(\Gamma_{j+1})$.
      Otherwise, $n>1$, and thus
      \begin{align*}
        |C_1\sqcap\dots\sqcap C_n|\cdot|D|
        &= (|C_1|+\dots+|C_n|+(n-1))\cdot|D| \\
        &> |C_1|\cdot|D|+\dots+|C_n|\cdot|D|.
      \end{align*}
    \item \rulefont{Atomic Decomposition:} It suffices to consider Case~e)
      since Case~a) is impossible and the other cases are trivial. Then
      $\sub=\exists r.C'\nsqsubseteq^?\exists r.D'$, and we get
      \[|\exists r.C'|\cdot|\exists r.D'|=(|C'|+1)\cdot(|D'|+1)>|C'|\cdot|D'|\,.\]
    \item \rulefont{Flattening Right-Ground Dissubsumptions:} Then
      $\sub=X\nsqsubseteq^?\exists r.D$ is replaced by
      $X\nsqsubseteq^?\exists r.X_D$ and $D\sqsubseteq^?X_D$. To the
      dissubsumption, no further rule is applicable, and hence it does not
      count towards $c(\Gamma_j)$. Regarding the subsumption, we have
      \[|X|\cdot|\exists r.D|=|D|+1>|D|=|D|\cdot|X_D|\,.\]
    \item \rulefont{Flattening Left-Ground Subsumptions:} Then the subsumption
      \sub is of the form \[C_1\sqcap\dots\sqcap C_n\sqcap
        \exists r_1.D_1\sqcap\dots\sqcap\exists r_m.D_m\sqsubseteq^?X\] and only
      to the subsumptions $D_1\sqsubseteq^?X_{D_1}$, \dots,
      $D_m\sqsubseteq^?X_{D_m}$ this rule may be applicable again.
      But we have
      \begin{align*}
        |C_1\sqcap{}&\dots\sqcap C_n\sqcap
              \exists r_1.D_1\sqcap\dots\sqcap\exists r_m.D_m|\cdot|X| \\
        &= |C_1|+\dots+|C_n|+|\exists r_1.D_1|+\dots+|\exists r_m.D_m|+(n+m-1)
          \displaybreak[0] \\
        &\ge |\exists r_1.D_1|+\dots+|\exists r_m.D_m|
          \displaybreak[0] \\
        &> |D_1|+\dots+|D_m| \\
        &= |D_1|\cdot|X_{D_1}|+\dots+|D_m|\cdot|X_{D_m}|.
      \end{align*}
    \item \rulefont{Solving Left-Ground Dissubsumptions:} Then
      $\sub=C_1\sqcap\dots\sqcap C_n\nsqsubseteq^?X$ and to a generated
      subsumption of the form $X\sqsubseteq^?A$ or $X\sqsubseteq^?\exists r.Z$
      no further rule is applicable. If $n=0$, then no further dissubsumptions
      are generated, and thus $c(\Gamma_j)>c(\Gamma_{j+1})$. Otherwise, we
      denote by $|\sub_i|$ the size of the dissubsumption resulting from
      applying \rulefont{Atomic Decomposition} to
      $C_i\nsqsubseteq^?\exists r.Z$, $1\le i\le n$, where we consider this
      number to be~$0$ if the dissubsumption was simply discarded (\cf
       Cases~b)--d) of \rulefont{Atomic Decomposition}).

      If $|\sub_i|=0$, we obtain $|C_i|\ge 1>0=|\sub_i|$. But also in Case~e),
      we have $C_i=\exists r.C_i'$, and thus
      $|C_i|=|C_i'|+1=|C_i'|\cdot|Z|+1>|\sub_i|$. Hence, we get
      \begin{align*}
        |C_1\sqcap\dots\sqcap C_n|\cdot|X|
        &= |C_1|+\dots+|C_n|+(n-1) \\
        &\ge |C_1|+\dots+|C_n| \\
        &> |\sub_1|+\dots+|\sub_n|,
      \end{align*}
      and thus again $c(\Gamma_j)>c(\Gamma_{j+1})$.
      \qed
  \end{itemize}

\noindent Note that the rule \rulefont{Solving Left-Ground Dissubsumptions} is not
limited to non-flat dissubsumptions, and thus the algorithm completely
eliminates all left-ground dissubsumptions from~$\Gamma$.
It is also easy to see that, if the algorithm is successful, then the resulting
disunification problem~$\Gamma$ is flat.
We now prove that this nondeterministic procedure is correct in the following
sense.

\begin{lem}
\label{lem:dismatching}
  The dismatching problem $\Gamma_0$ is solvable iff there is a successful run
  of Algorithm~\ref{alg:dismatching-reduction} such that the resulting flat
  disunification problem~$\Gamma$ has a local solution.
\end{lem}
\begin{proof}
  For soundness (\ie the ``if'' direction),
  let $\sigma$ be the local solution of~$\Gamma$ and consider the run of
  Algorithm~\ref{alg:dismatching-reduction} that produced~$\Gamma$. It is easy
  to show by induction on the reverse order in which the rules have been
  applied that $\sigma$ solves all subsumptions that have been considered.
  Indeed, this follows from simple applications of
  Lemmata~\ref{lem:sub}--\ref{lem:dissub-atoms} and the properties of
  subsumption.
  This implies that $\sigma$ is also a solution of~$\Gamma_0$.

  Showing completeness (\ie the ``only if'' direction) is a little more
  involved.
  Let $\gamma$ be a solution of~$\Gamma_0$. We guide the rule applications of
  Algorithm~\ref{alg:dismatching-reduction} and extend~$\gamma$ to the newly
  introduced variables in such a way to maintain the invariant that
  ``\emph{$\gamma$~solves all (dis)subsumptions of~$\Gamma$}''. This obviously
  holds after the initialization $\Gamma:=\Gamma_0$.
  Afterwards, we will use~$\gamma$ to define a local solution of~$\Gamma$.

  Consider a (dis)subsumption~$\sub\in\Gamma$ (which is solved by~$\gamma$) to
  which one of the rules of Figures~\ref{fig:dismatching-decomp}
  and~\ref{fig:dismatching-flatten} is applicable. We make a case distinction
  on which rule is to be applied:
  \begin{itemize}
    \item \rulefont{Right Decomposition:} Then \sub is of the form
      $C_1\sqcap\dots\sqcap C_n\nsqsubseteq^?D_1\sqcap\dots\sqcap D_m$ for
      $m\neq 1$. Since $\gamma(C_1\sqcap\dots\sqcap C_n) \nsqsubseteq
      \gamma(D_1\sqcap\dots\sqcap D_m)$, by applying Lemma~\ref{lem:dissub}
      twice, we can find an index $i\in\{1,\dots,m\}$ such that
      $\gamma(C_1\sqcap\dots\sqcap C_n) \nsqsubseteq \gamma(D_i)$. Thus, we can
      choose this index in the rule application in order to satisfy the
      invariant.
    \item \rulefont{Left Decomposition:} Then \sub is of the form
      $C_1\sqcap\dots\sqcap C_n\nsqsubseteq^?D$, where $n\neq 1$ and $D$ is a
      non-variable atom. This means that $\gamma(D)$ is also an atom, and thus
      by Lemma~\ref{lem:dissub} we know that $\gamma(C_i)\nsqsubseteq\gamma(D)$
      holds for all $i\in\{1,\dots,n\}$, as required.
    \item \rulefont{Atomic Decomposition:} Then \sub is of the form
      $C\nsqsubseteq^?D$ for two non-variable atoms~$C$ and~$D$. Since
      $\gamma(C)\nsqsubseteq\gamma(D)$, Case a) cannot apply. If one of the
      Cases b)--d) applies, then \sub is simply removed from~$\Gamma$ and there
      is nothing to show. Otherwise, we have $D=\exists r.D'$ and
      $C=\exists r.C'$, and the new dissubsumption $C'\nsqsubseteq^?D'$ is
      added to~$\Gamma$. Moreover, we have $\gamma(C)=\exists r.\gamma(C')$
      and~$\gamma(D)=\exists r.\gamma(D')$, and thus by
      Lemma~\ref{lem:dissub-atoms} we know that
      $\gamma(C')\nsqsubseteq\gamma(D')$.
    \item \rulefont{Flattening Right-Ground Dissubsumptions:} Then \sub is of
      the form $X\nsqsubseteq^?\exists r.D$. By defining $\gamma(X_D):=D$,
      $\gamma$ solves $X\nsqsubseteq^?\exists r.X_D$ and $D\sqsubseteq^?X_D$.
    \item \rulefont{Flattening Left-Ground Subsumptions:} Then the
      subsumption~\sub is of the form \[C_1\sqcap\dots\sqcap C_n\sqcap
        \exists r_1.D_1\sqcap\dots\sqcap\exists r_m.D_m\sqsubseteq^?X\,,\] where
      all $D_1,\dots,D_m$ are ground. If we extend~$\gamma$ by defining
      $\gamma(X_{D_i}):=D_i$ for all $i\in\{1,\dots,m\}$, then this obviously
      satisfies the new subsumptions $D_1\sqsubseteq^?X_{D_1}$, \dots,
      $D_m\sqsubseteq^?X_{D_m}$, and $C_1\sqcap\dots\sqcap C_n\sqcap
        \exists r_1.X_{D_1}\sqcap\dots\sqcap\exists r_m.X_{D_m}\sqsubseteq^?X$
      by our assumption that $\gamma$ solves~\sub.
    \item \rulefont{Solving Left-Ground Dissubsumptions:} Then the
      dissubsumption~\sub is of the form
      \[C_1\sqcap\dots\sqcap C_n\nsqsubseteq^?X\,,\] where $X$ is a variable and
      $C_1,\dots,C_n$ are ground atoms. By Lemma~\ref{lem:dissub}, there must
      be a ground top-level atom~$D$ of~$\gamma(X)$ such that
      $C_1\sqcap\dots\sqcap C_n\nsqsubseteq D$, \ie $C_1\nsqsubseteq D$, \dots,
      $C_n\nsqsubseteq D$. If $D$ is a concept constant, we can choose this in
      the rule application since we know that $\gamma(X)\sqsubseteq D$.
      Otherwise, we have $D=\exists r.D'$. By extending~$\gamma$ to
      $\gamma(Z):=D'$, we ensure that $X\sqsubseteq^?\exists r.Z$,
      $C_1\nsqsubseteq^?\exists r.Z$, \dots $C_n\nsqsubseteq^?\exists r.Z$ are
      solved by~$\gamma$.
      The remaining claim follows as for \rulefont{Atomic Decomposition} above.
  \end{itemize}
  Once no more rules can be applied, we obtain a flat disunification
  problem~$\Gamma$ of which the extended substitution~$\gamma$ is a (possibly
  non-local) solution.
  To obtain a local solution, we denote by~\At, \Var, and~\NV the sets as
  defined in Section~\ref{sec:disunification} and define the assignment~$S$
  induced by~$\gamma$ as in~\cite{BaMo-LPAR10}:
  \[ S_X := \{ D\in\NV \mid \gamma(X)\sqsubseteq\gamma(D) \}, \]
  for all (old and new) variables~$X\in\Var$. It was shown
  in~\cite{BaMo-LPAR10} that $S$ is acyclic and the substitution~$\sigma_S$
  solves all subsumptions in~$\Gamma$.%
\footnote{More precisely, it was shown that $\gamma$ induces a satisfying
valuation of a SAT problem, which in turn induces the solution~$\sigma_S$
above. For details, see~\cite{BaMo-LPAR10} or Sections~\ref{sec:sat-soundness}
and~\ref{sec:sat-completeness}.}
  Furthermore, it is easy to show that $\gamma(C)\sqsubseteq\sigma_S(C)$ holds
  for all concept terms~$C$.

  Since $\Gamma$ contains no left-ground dissubsumptions anymore, it remains to
  show that $\sigma_S$ solves all remaining right-ground dissubsumptions
  in~$\Gamma$ and all flat dissubsumptions created by an application of the
  rule \rulefont{Flattening Right-Ground Dissubsumptions}.
  Consider first any flat right-ground dissubsumption $X\nsqsubseteq^?D$
  in~$\Gamma$. We have already shown that $\gamma(X)\nsqsubseteq D$ holds.
  Since $\gamma(X)\sqsubseteq\sigma_S(X)$, by the transitivity of subsumption
  $\sigma_S(X)\sqsubseteq D$ cannot hold, and thus $\sigma_S$ also solves the
  dissubsumption.

  Consider now a dissubsumption $X\nsqsubseteq^?\exists r.X_D$ that was created
  by an application of the rule \rulefont{Flattening Right-Ground
  Dissubsumptions} to $X\nsqsubseteq^?\exists r.D$. By the same argument as
  above, from $\gamma(X)\nsqsubseteq\exists r.D$ we can derive that
  $\sigma_S(X)\nsqsubseteq\exists r.D$ holds.
  We now show that $\sigma_S(X_D)\sqsubseteq D$ holds, which implies that
  $\sigma_S(\exists r.X_D)\sqsubseteq\exists r.D$, and thus by the transitivity
  of subsumption it cannot be the case that
  $\sigma_S(X)\sqsubseteq\sigma_S(\exists r.X_D)$, which concludes the proof by
  showing that $\sigma_S$ solves~$\Gamma$.

  We show that $\sigma_S(X_C)\sqsubseteq C$ holds for all variables~$X_C$ for
  which a subsumption $C\sqsubseteq^?X_C$ was introduced by a
  \rulefont{Flattening} rule.
  We prove this claim by induction on the \emph{role depth} of~$C$, which is
  the maximum nesting depth of existential restrictions occurring in it.
  Let $C_1,\dots,C_n$ be the top-level atoms of~$C$. Then $\Gamma$ contains a
  flat subsumption $C_1'\sqcap\dots\sqcap C_n'\sqsubseteq^?X_C$, where
  $C_i=C_i'$ if $C_i$ is flat, and $C_i=\exists r.D_i$ and
  $C_i'=\exists r.X_{D_i}$ otherwise. Since the role depth of each such~$D_i$
  is strictly smaller than that of~$C$, by induction we know that
  $\sigma_S(X_{D_i})\sqsubseteq D_i$, and thus
  $\sigma_S(C_1'\sqcap\dots\sqcap C_n') \sqsubseteq
    C_1\sqcap\dots\sqcap C_n = C$ by Lemma~\ref{lem:sub}.
  Furthermore, for all $i\in\{1,\dots,n\}$ we have
  $\gamma(X_C)=C\sqsubseteq C_i=\gamma(C_i')$ and $C_i'\in\NV$. Thus,
  $C_i'\in S_{X_C}$ by the definition of~$S$. The definition of~$\sigma_S$ now
  yields that
  $\sigma_S(X_C)\sqsubseteq \sigma_S(C_1'\sqcap\dots\sqcap C_n')\sqsubseteq C$
  (see Section~\ref{sec:local-disunification}).
\end{proof}
The disunification problem of Example~\ref{exa:disunification-not-local}
is in fact a dismatching problem. Applying
Algorithm~\ref{alg:dismatching-reduction} to this problem, we can use the rule
\rulefont{Solving Left-Ground Dissubsumptions} to replace $\top\nsqsubseteq^?Y$
with $Y\sqsubseteq^? \exists r.Z$.
The presence of the new atom $\exists r.Z$ makes the solution~$\sigma$
introduced in Example~\ref{exa:disunification-not-local} local.

Together with Fact~\ref{fact:np} and the \NP-hardness of unification
in~\EL~\cite{BaMo-LMCS10}, this shows the following complexity result.

\begin{thm}
\label{thm:dismatching-np}
  Dismatching in \EL is \NP-complete.
\qed
\end{thm}
Additionally, one can see from the proof of Lemma~\ref{lem:dismatching} that
any local solution of the constructed disunification problem~$\Gamma$ is also a
solution of the original problem~$\Gamma_0$.
Hence, if we are interested in actually computing solutions of~$\Gamma_0$ in
order to show them to the user, we can collect the solutions of the flat
problems~$\Gamma$ produced by the successful runs of
Algorithm~\ref{alg:dismatching-reduction}.

\section{A goal-oriented algorithm for local disunification}
\label{sec:rules}

In this section, we present a sound and complete algorithm that provides a more
goal-directed way to solve local disunification problems than blindly guessing
an assignment as described in Section~\ref{sec:dismatching}.
The approach is based on transformation rules that are applied to subsumptions
and dissubsumptions in order to derive a local solution.
To solve the \emph{subsumptions}, we reuse the rules of the goal-oriented
algorithm for unification in~\EL~\cite{BaMo-LMCS10,BaBM-DL12}, which produces
only local unifiers.
Since any local solution of the disunification problem is in particular a local
unifier of the subsumptions in the problem, one might think that it is then
sufficient to check whether any of the produced unifiers also solves the
dissubsumptions.
This would not be complete, however, since the goal-oriented algorithm for
unification does \emph{not} produce  \emph{all} local unifiers.  For this
reason, we have additional rules for solving the dissubsumptions. Both rule
sets contain (deterministic) \emph{eager} rules that are applied with the
highest priority, and \emph{nondeterministic} rules that are only applied if no
eager rule is applicable. The goal of the eager rules is to enable the
algorithm to detect obvious contradictions as early as possible in order to
reduce the number of nondeterministic choices it has to make.

Let now $\Gamma_0$ be the flat disunification problem for which we want to
decide local solvability, and let the sets \At, \Var, and \NV be defined as in
Section~\ref{sec:disunification}.
We assume without loss of generality that the dissubsumptions in~$\Gamma_0$
have only a single atom on the right-hand side. If this is not the case, it can
easily be achieved by exhaustive application of the nondeterministic rule
\rulefont{Right Decomposition} (see Figure~\ref{fig:dismatching-decomp})
without affecting the complexity of the overall procedure.

Starting with $\Gamma_0$, the algorithm maintains a current disunification
problem~$\Gamma$ and a current acyclic assignment~$S$, which initially assigns
the empty set to all variables.
In addition, for each subsumption or dissubsumption in~$\Gamma$, it maintains
the information on whether it is \emph{solved} or not.
Initially, all subsumptions of $\Gamma_0$ are unsolved, except those with a
variable on the right-hand side, and all dissubsumptions in $\Gamma_0$ are
unsolved, except those with a variable on the left-hand side and a non-variable
atom on the right-hand side.

Subsumptions of the form $C_1 \sqcap \dots \sqcap  C_n \sqsubseteq^? X$ and
dissubsumptions of the form $X \nsqsubseteq^? D$, for a non-variable atom~$D$,
are called \emph{initially solved}.
Intuitively, they only specify constraints on the assignment~$S_X$. More
formally, this intuition is captured by the process of \emph{expanding}
$\Gamma$ w.r.t.\ the variable~$X$, which performs the following actions:
\begin{itemize}
  \item every initially solved subsumption $\sub \in \Gamma$ of the form
    $C_1 \sqcap \dots \sqcap  C_n \sqsubseteq^? X$ is expanded by adding
    the subsumption $C_1 \sqcap \dots \sqcap C_n \sqsubseteq^?E$ to $\Gamma$
    for every $E\in S_X$, and
  \item every initially solved dissubsumption $X \not\sqsubseteq^? D \in 
  \Gamma$ is expanded
    by adding $E \not\sqsubseteq^? D$ to $\Gamma$ for every $E\in S_X$.
\end{itemize}
A (non-failing) application of a rule of our algorithm does the following:
\begin{itemize}
\item it solves exactly one unsolved subsumption or dissubsumption,
\item it may extend the current assignment $S$ by adding elements of \NV to
  some set~$S_X$,
\item it may introduce new flat subsumptions or dissubsumptions built from 
elements of \At, and
\item it keeps $\Gamma$ expanded \wrt all variables~$X$.
\end{itemize}
Subsumptions and dissubsumptions are only added by a rule application or by 
expansion if they are not already present in $\Gamma$.
If a new subsumption or dissubsumption is added to $\Gamma$, it is marked as 
unsolved, unless it is initially solved (because of its form).
Solving subsumptions and dissubsumptions is mostly independent, except for
expanding~$\Gamma$, which can add new unsolved subsumptions and dissubsumptions
at the same time, and may be triggered by solving a subsumption or a
dissubsumption.

The rules of our algorithm are depicted in
Figures~\ref{fig:eager-rules} and~\ref{fig:nondet-rules}.
The rules dealing with subsumptions are essentially the same as
in~\cite{BaBM-DL12}; note that several of these may be applicable to the same
subsumption.
In the rule \rulefont{Local Extension}, the left-hand side of \sub may be a
variable, and then \sub is of the form $Y \nsqsubseteq^? X$. This
dissubsumption is not initially solved, because $X$ is not a non-variable atom.
\begin{figure}[tb]
  \centering
  \begin{algrule}[Eager Ground Solving]
    \condition This rule applies to
      $\sub=C_1\sqcap\dots\sqcap C_n \bowtie^? D$ with
      ${\bowtie}\in\{\sqsubseteq,\nsqsubseteq\}$ if \sub is ground.

    \action If $C_1\sqcap\dots\sqcap C_n\bowtie D$, then mark \sub as
      \emph{solved}; otherwise, \emph{fail}.
  \end{algrule}

  \rulespace
  \begin{algrule}[Eager Solving]
    \condition This rule applies to
      $\sub=C_1\sqcap\dots\sqcap C_n \bowtie^? D$ with
      ${\bowtie}\in\{\sqsubseteq,\nsqsubseteq\}$ if there is an index
      $i \in \{1, \dots, n\}$ such that $C_i = D$ or $C_i$ is a variable with
      $D \in S_{C_i}$.

    \action If ${\bowtie}={\sqsubseteq}$, then mark \sub as \emph{solved};
      otherwise, \emph{fail}.
  \end{algrule}

  \rulespace
  \begin{algrule}[Eager Extension]
    \condition This rule applies to
      $\sub=C_1\sqcap\dots\sqcap C_n \sqsubseteq^? D \in \Gamma$ if there is an
      index $i \in \{1, \dots, n\}$ such that $C_i$ is a variable and
      $\{C_1,\dots, C_n\}\setminus \{C_i\} \subseteq S_{C_i}$.

    \action Add $D$ to $S_{C_i}$. If this makes $S$ cyclic, then \emph{fail}.
      Otherwise, expand~$\Gamma$ w.r.t.~$C_i$ and mark \sub as \emph{solved}.
  \end{algrule}

  \rulespace
  \begin{algrule}[Eager Top Solving]
    \condition This rule applies to
      $\sub=C\nsqsubseteq^?\top\in\Gamma$.

    \action \emph{Fail}.
  \end{algrule}

  \rulespace
  \begin{algrule}[Eager Left Decomposition]
    \condition This rule applies to
      $\sub=C_1\sqcap\dots\sqcap C_n \not\sqsubseteq^? D \in \Gamma$ if
      $n\neq 1$ and $D$ is a non-variable atom.

    \action Mark \sub as \emph{solved} and, for each $i \in \{1, \dots, n\}$,
      add $C_i \not\sqsubseteq^? D$ to~$\Gamma$ and expand~$\Gamma$
      w.r.t.~$C_i$ if $C_i$ is a variable.
  \end{algrule}

  \rulespace
  \begin{algrule}[Eager Atomic Decomposition]
    \condition This rule applies to $\sub=C \not\sqsubseteq^?D\in\Gamma$
      if $C$ and $D$ are non-variable atoms.

    \action Apply the first case that matches~\sub:
      \begin{enumerate}[label=\alph*)]
        \item if $C$ and $D$ are ground and $C\sqsubseteq D$, then \emph{fail};
        \item if $C$ and $D$ are ground and $C\nsqsubseteq D$, then mark \sub
          as \emph{solved};
        \item if $C$ or $D$ is a constant, then mark \sub as
          \emph{solved};
        \item if $C = \exists r. C'$ and $D = \exists s.D'$ with $r\neq s$,
          then mark \sub as \emph{solved};
        \item if $C = \exists r. C'$ and $D = \exists r.D'$, then add
           $C' \not\sqsubseteq^? D'$ to~$\Gamma$, expand~$\Gamma$ w.r.t.~$C'$
           if $C'$ is a variable and $D'$ is not a variable, and mark \sub as
           \emph{solved}.
      \end{enumerate}
  \end{algrule}
  \vspace*{-2ex}
  \caption{Eager rules for Algorithm~\ref{alg:rules-local-disunification}}
  \label{fig:eager-rules}
\end{figure}
\begin{figure}[tb]
  \centering
  \begin{algrule}[Decomposition]
    \condition This rule applies to
      $\sub= C_1 \sqcap \dots \sqcap C_n \sqsubseteq^? \exists s. D \in \Gamma$
      if there is an index $i \in \{1, \dots, n\}$ such that
      $C_i = \exists s.C$.

    \action Choose such an index~$i$, add $C \sqsubseteq^? D$ to~$\Gamma$,
      expand $\Gamma$ w.r.t.~$D$ if $D$ is a variable, and mark \sub as
      \emph{solved}.
  \end{algrule}

  \rulespace
  \begin{algrule}[Extension\vphantom{g}]
    \condition This rule applies to
      $\sub= C_1 \sqcap \dots \sqcap C_n \sqsubseteq^? D \in \Gamma$ if there
      is an index $i \in \{1, \dots, n\}$ such that $C_i$ is a variable.

    \action Choose such an index~$i$ and add~$D$ to~$S_{C_i}$. If this makes
      $S$ cyclic, then \emph{fail}. Otherwise, expand~$\Gamma$ w.r.t.~$C_i$
      and mark \sub as \emph{solved}.
  \end{algrule}

  \rulespace
  \begin{algrule}[Local Extension\vphantom{g}]
    \condition This rule applies to
     $\sub= C \nsqsubseteq^? X \in \Gamma$ if $X$ is a variable.

    \action Choose a non-variable atom~$D$ and add $D$ to~$S_X$. If this makes
      $S$ cyclic, then \emph{fail}. Otherwise, add $C \nsqsubseteq^? D$
      to~$\Gamma$, expand~$\Gamma$ w.r.t.~$X$, expand~$\Gamma$ w.r.t.~$C$ if
      $C$ is a variable, and mark \sub as \emph{solved}.
  \end{algrule}
  \vspace*{-2ex}
  \caption{Nondeterministic rules for
    Algorithm~\ref{alg:rules-local-disunification}}
  \label{fig:nondet-rules}
\end{figure}

\begin{algo}
\label{alg:rules-local-disunification}
Let $\Gamma_0$ be a flat disunification problem. We initialize
$\Gamma := \Gamma_0$ and $S_X := \emptyset$  for all variables~$X$.
While $\Gamma$ contains an unsolved element, do the following:
\begin{enumerate}
  \item \rulefont{Eager rule application:} If any eager rules
    (Figure~\ref{fig:eager-rules}) are applicable to some unsolved element
    $\sub\in\Gamma$, apply an arbitrarily chosen one to~\sub.
    If the rule application fails, return ``failure''.
  \item \rulefont{Nondeterministic rule application:} If no eager rule is
    applicable, let \sub be an unsolved subsumption or dissubsumption
    in~$\Gamma$. If one of the nondeterministic rules
    (Figure~\ref{fig:nondet-rules}) applies to~\sub, choose one and apply it.
    If none of these rules apply to~\sub or the rule application fails, return
    ``failure''.
\end{enumerate}
Once all elements of~$\Gamma$ are solved, return the substitution~$\sigma_S$
that is induced by the current assignment.
\end{algo}
As with Algorithm~\ref{alg:dismatching-reduction}, the choice which
(dis)subsumption to consider next and which eager rule to apply is
\emph{don't care} nondeterministic, while the choice of which nondeterministic
rule to apply and the choices inside the rules are \emph{don't know}
nondeterministic.
Each of these latter choices may result in a different solution~$\sigma_S$.

\subsection{Termination}

\begin{lem}
Every run of Algorithm~\ref{alg:rules-local-disunification} terminates in time
polynomial in the size of~$\Gamma_0$.
\end{lem}
\proof
Each rule application solves one subsumption or dissubsumption.
We show that only polynomially many subsumptions and dissubsumptions are
produced during a run of the algorithm, and thus there can be only polynomially
many rule applications during one run of the algorithm.

A new subsumption or dissubsumption may be created only by an application of
the rules \rulefont{Decomposition}, \rulefont{Eager Left Decomposition}, or
\rulefont{Eager Atomic Decomposition}, and then it is of the form
$C\sqsubseteq^?D$ or $C\nsqsubseteq^?D$, with $C,D\in\At$.
Obviously, there are only polynomially many such (dis)subsumptions.

Now, we consider (dis)subsumptions created by expanding~$\Gamma$. They can have
the following forms, where $D,E\in\NV$:
\begin{enumerate}
\item $C_1 \sqcap \dots \sqcap C_n \sqsubseteq^? E$, for $C_1 \sqcap 
\dots 
\sqcap C_n \sqsubseteq^? X$ in $\Gamma$,
\item  $E \nsqsubseteq^? D$, for $X \nsqsubseteq^? D$ in $\Gamma$.
\end{enumerate}
Dissubsumptions of the type~(2) are also of the form described above.
For the subsumptions of type~(1), note that $C_1 \sqcap \dots \sqcap C_n$
is either the left-hand side of a subsumption from the original
problem~$\Gamma_0$, or was created by a \rulefont{Decomposition} rule, in
which case we have $n=1$.
Thus, there can also be at most polynomially many subsumptions of the first
type.

Finally, each rule application takes at most polynomial time.
\qed

\subsection{Soundness}

Assume that a run of the algorithm terminates with success, \ie all
subsumptions and dissubsumptions are solved.
Let $\hat{\Gamma}$ be the set of all subsumptions and dissubsumptions produced 
by this run, $S$ be the final assignment, and $\sigma_S$ the induced
substitution (see Section~\ref{sec:disunification}).
Observe that the algorithm never removes elements from the current
disunification problem, but only marks them as solved, and hence $\hat{\Gamma}$
contains~$\Gamma_0$.
To show that $\sigma_S$ solves~$\hat{\Gamma}$, and thus~$\Gamma_0$, we use
induction on the following order on (dis)subsumptions.

\begin{defi}
  Consider any (dis)subsumption~\sub of the form
  $C_1\sqcap\dots\sqcap C_n\sqsubseteq^?C_{n+1}$ or
  $C_1\sqcap\dots\sqcap C_n\nsqsubseteq^?C_{n+1}$ in~$\hat{\Gamma}$.
  \begin{itemize}
    \item We define $m(\sub) := (m_1(\sub), m_2(\sub))$, where
      \begin{itemize}
        \item $m_1(\sub):=\{X_1,\dots,X_m\}$ is the multiset containing all
          occurrences of variables in the concept terms
          $C_1,\dots,C_n,C_{n+1}$ (and hence $m_1(\sub)=\emptyset$ if \sub is
          ground);
        \item $m_2(\sub):=|\sub|$ is the size of~\sub, \ie the number of
          symbols in~\sub (see the proof of
          Lemma~\ref{lem:dismatching-termination}).
      \end{itemize}
    \item The strict partial order $\succ$ on such pairs is the lexicographic
      order, where the second components are compared \wrt the usual order
      on natural numbers, and the first components are compared \wrt the
      multiset extension of~$>_S$~\cite{BaNi-99}.
    \item We extend $\succ$ to $\hat{\Gamma}$ by setting $\sub_1\succ\sub_2$
      iff $m(\sub_1)\succ m(\sub_2)$.
  \end{itemize}
\end{defi}
Since multiset extensions and lexicographic products of well-founded strict
partial orders are again well-founded~\cite{BaNi-99}, $\succ$ is a well-founded
strict partial order on $\hat{\Gamma}$.

\begin{lem}
  The substitution~$\sigma_S$ is a solution of $\hat{\Gamma}$, and thus also of
  its subset~$\Gamma_0$.
\end{lem}
\proof
Consider a (dis)subsumption $\sub\in\hat{\Gamma}$ and assume that $\sigma_S$
solves all $\sub'\in\hat{\Gamma}$ with $\sub'\prec\sub$.
Since \sub is solved, either it has been solved by a rule application or it was
initially solved.

If \sub is a dissubsumption that is initially solved, then
$\sub = X \nsqsubseteq^? D$, where $X$ is a variable and $D$ a non-variable
atom. By expansion, for every $E\in S_X$, we have
$\sub_E=E\nsqsubseteq^?D\in\hat{\Gamma}$. We know that $\sub\succ\sub_E$,
because $E$ may only contain a variable strictly smaller than~$X$, and thus
$m_1(\sub)>m_1(\sub_E)$. Hence by induction, $\sigma_S$ solves all
dissubsumptions~$\sub_E$ with $E\in S_X$.
Since the top-level atoms of~$\sigma_S(X)$ are exactly those of the
form~$\sigma_S(E)$ for $E\in S_X$, by Lemma~\ref{lem:dissub} we know that
$\sigma_S$ also solves~\sub.

If \sub is a subsumption that is initially solved, then
$\sub=C_1\sqcap\dots\sqcap C_n\sqsubseteq^?X$, where $X$ is a variable. By
expansion, for every $E\in S_X$, there is a subsumption
$\sub_E=C_1\sqcap\dots\sqcap C_n\sqsubseteq^?E$ in~$\hat{\Gamma}$. We have
$\sub_E\prec\sub$ since $m_1(\sub_E)<m_1(\sub)$, for every $E\in S_X$.
Hence, by induction all subsumptions~$\sub_E$ are solved by~$\sigma_S$.
By the definition of $\sigma_S(X)$ and Lemma~\ref{lem:sub}, $\sigma_S$
solves~\sub.

If \sub was solved by a rule application, we consider which rule was applied.
\begin{itemize}
  \item \rulefont{Eager Ground Solving:} Then \sub is ground and holds under
    any substitution.
  \item \rulefont{Eager Solving:} Since this rule fails for all
    dissubsumptions to which it is applicable, but we assumed that the run was
    successful, we have $\sub=C_1\sqcap\dots\sqcap C_n\sqsubseteq^?D$ and
    $\sigma_S(D)$ occurs on the top-level of
    $\sigma_S(C_1)\sqcap\dots\sqcap\sigma_S(C_n)$. Hence, $\sigma_S$ solves
    the subsumption.
  \item \rulefont{(Eager) Extension:} Then
    $\sub=X\sqcap C_1\sqcap\dots\sqcap C_n\sqsubseteq^?D$ for a variable~$X$
    and $D\in S_X$. By the definition of~$\sigma_S$, we have
    $\sigma_S(X)\sqsubseteq\sigma_S(D)$ and thus $\sigma_S$ solves~\sub.
  \item \rulefont{Eager Top Solving:} This rule cannot have been applied since
    we assumed the run to be successful.
  \item \rulefont{Eager Left Decomposition:} Then either
    $\sub = C_1 \sqcap \dots \sqcap C_n \nsqsubseteq^?D$ with $n>1$, or
    $\sub=\top\nsqsubseteq^?D$, for a non-variable atom~$D$. In the latter
    case, $\sigma_S$ solves~\sub by Lemma~\ref{lem:dissub}. In the former
    case, for each $i\in\{1,\dots,n\}$ we have
    $\sub_i := C_i \nsqsubseteq^?D\in\hat{\Gamma}$. Notice that
    $m_1(\sub) \ge  m_1(\sub_i)$ and $m_2(\sub) > m_2(\sub_i)$ and hence
    $\sub \succ \sub_i$. Thus, by induction we have that
    $\sigma_S(C_i)\nsqsubseteq\sigma_S(D)$. By applying Lemma~\ref{lem:dissub}
    twice, we conclude that
    $\sigma_S(C_1) \sqcap \dots \sqcap \sigma_S(C_n) \nsqsubseteq 
      \sigma_S(D)$.
  \item \rulefont{Eager Atomic Decomposition:} Then
    $\sub = C \nsqsubseteq^? D$, where $C$ and $D$ are non-variable atoms.
    Since we assume that the run was successful, Case~a) cannot apply.
    In Cases b)--d), $\sigma_S$ must solve~\sub by
    Lemma~\ref{lem:dissub-atoms}.
    Finally, in Case~e), we have $C=\exists r.C'$, $D=\exists r.D'$, and
    $\sub'=C'\nsqsubseteq^?D'\in\hat{\Gamma}$. Notice that
    $\sub\succ\sub'$, because $m_1(\sub)=m_1(\sub')$ and
    $m_2(\sub)>m_2(\sub')$. Hence, by induction we get
    $\sigma_S(C')\nsqsubseteq\sigma_S(D')$ and thus
    $\sigma_S(C)\nsqsubseteq\sigma_S(D)$ by Lemma~\ref{lem:dissub-atoms}.
  \item \rulefont{Decomposition:} Then
    $\sub=C_1\sqcap\dots\sqcap C_n\sqsubseteq^?\exists s.D$ with
    $C_i=\exists s.C$ for some $i\in\{1,\dots,n\}$ and we have
    $\sub'=C\sqsubseteq^?D\in\hat{\Gamma}$. We know that $\sub'\prec\sub$,
    because $m_1(\sub')\le m_1(\sub)$ and $m_2(\sub')<m_2(\sub)$. By
    induction, we get $\sigma_S(C)\sqsubseteq\sigma_S(D)$, and hence
    $\sigma_S$ solves~\sub.
  \item \rulefont{Local Extension:} Then
    $\sub=C_1\sqcap\dots\sqcap C_n\nsqsubseteq^?X$
    and there is a non-variable atom $D\in S_X$ such that
    $\sub'=C_1\sqcap\dots\sqcap C_n\nsqsubseteq^?D\in\hat{\Gamma}$. We have
    $\sub \succ \sub'$, because $D$ may only contain a variable strictly
    smaller than $X$, and thus $m_1(\sub) > m_1(\sub')$. Hence by induction,
    $\sigma$ solves~$\sub'$. Since $\sigma_S(D)$ is a top-level atom
    of~$\sigma_S(X)$, $\sigma_S$ solves~\sub by Lemma~\ref{lem:dissub}.
    \qed
\end{itemize}

\subsection{Completeness}

Assume now that $\Gamma_0$ has a local solution~$\sigma$.
We show that $\sigma$ can guide the choices of
Algorithm~\ref{alg:rules-local-disunification} to obtain a local
solution~$\sigma'$ of~$\Gamma_0$ such that, for every variable $X$, we have
$\sigma(X) \sqsubseteq \sigma'(X)$.
The following invariants will be maintained throughout the run of the algorithm
for the current set of (dis)subsumptions~$\Gamma$ and the current
assignment~$S$:
\begin{enumerate}[label=(\Roman*), widest=II]
  \item\label{i}
    $\sigma$ is a solution of~$\Gamma$.
  \item\label{ii}
    For each $D\in S_X$, we have $\sigma(X)\sqsubseteq\sigma(D)$.
\end{enumerate}
By Lemma~\ref{lem:sub}, chains of the form
$\sigma(X_1)\sqsubseteq\sigma(\exists r_1.X_2)$, \dots
$\sigma(X_{n-1})\sqsubseteq\sigma(\exists r_{n-1}.X_n)$ with $X_1=X_n$ are
impossible, and thus invariant~\ref{ii} implies that $S$ is acyclic.
Hence, if extending $S$ during a rule application preserves this invariant,
this extension will not cause the algorithm to fail.

\begin{lem}
  The invariants are maintained by the operation of expanding~$\Gamma$.
\end{lem}
\proof
  Since expansion does not affect the assignment $S$, we have to check only
  invariant~\ref{i}.
  Consider a subsumption $\sub=C_1\sqcap\dots\sqcap C_n\sqsubseteq^?X$
  in~$\Gamma$, for which a new subsumption
  $\sub_E=C_1\sqcap\dots\sqcap C_n\sqsubseteq^?E$ is created because
  $E\in S_X$.
  By the invariants, $\sigma$ solves~\sub and $\sigma(X)\sqsubseteq\sigma(E)$.
  Hence by transitivity of subsumption, $\sigma$ also solves $\sub_E$, \ie
  invariant~\ref{i} is satisfied after adding~$\sub_E$ to~$\Gamma$.

  For a dissubsumption $\sub=X\nsqsubseteq^?D\in\Gamma$ and $E\in S_X$, a new
  dissubsumption $\sub_E=E\nsqsubseteq^?D$ is created. Since $\sigma$
  solves~\sub and $\sigma(X)\sqsubseteq\sigma(E)$ by invariant~\ref{ii}, we
  have $\sigma(E)\nsqsubseteq\sigma(D)$ by transitivity of subsumption, \ie
  $\sigma$ solves~$\sub_E$.
\qed

Now we show that if the invariants are satisfied, the eager rules maintain the
invariants and do not lead to failure.

\begin{lem}
  The application of an eager rule never fails and maintains the invariants.
\end{lem}
\proof
  There are six eager rules to consider:
  \begin{itemize}
    \item \rulefont{Eager Ground Solving:} By invariant~\ref{i}, $\sigma$
      solves all ground (dis)subsumptions in~$\Gamma$, and thus they must be
      valid. Therefore the rule cannot fail, and obviously it preserves the
      invariants.
    \item \rulefont{Eager Solving:} The rule does not affect the invariants. It
      could fail only in the case that $\Gamma$ contains a dissubsumption
      $\sub=C_1\sqcap\dots\sqcap C_n\nsqsubseteq^?D$ for which there exists an
      index $i\in\{1,\dots,n\}$ such that $C_i=D$ or $C_i$ is a variable and
      $D\in S_{C_i}$. By invariant~\ref{i} and Lemma~\ref{lem:sub}, the former
      case is impossible. In the latter case, invariant~\ref{ii} similarly
      yields a contradiction to invariant~\ref{i}.
    \item \rulefont{Eager Extension:} Consider any
      $C_1 \sqcap \dots \sqcap C_m \sqsubseteq^? D \in \Gamma$ such that
      there is an index $i \in \{1, \dots, n\}$ with $C_i = X\in\Var$ and
      $\{C_1, \dots, C_m\}\setminus \{X\} \subseteq S_X$.
      By the invariants and Lemma~\ref{lem:sub}, we have
      $\sigma(X)\sqsubseteq\sigma(C_1)\sqcap\dots\sqcap\sigma(C_m)
        \sqsubseteq\sigma(D)$, and thus adding~$D$ to~$S_X$ maintains
      invariant~\ref{ii}. Therefore, the application of the rule does not cause
      $S$ to be cyclic, and does not fail. Invariant~\ref{i} is not affected by
      this rule.
    \item \rulefont{Eager Top Solving:} By invariant~\ref{i}, this rule will
      never be applied since $\sigma(C)\nsqsubseteq^?\top$ is impossible by
      Lemma~\ref{lem:dissub}.
    \item \rulefont{Eager Left Decomposition:} This rule never fails.
      Furthermore, $S$ is not affected by the rule, and hence invariant~\ref{i}
      is preserved. Finally, if $\sigma$ solves
      $C_1\sqcap\dots\sqcap C_n\nsqsubseteq^?D$, then it must also solve
      $C_i \nsqsubseteq^? D$ for each $i \in \{1, \dots, n\}$ by
      Lemma~\ref{lem:dissub}.
    \item \rulefont{Eager Atomic Decomposition:} Case~a) cannot apply since
      $\sigma$ is a solution of~$\Gamma$. Invariant~\ref{ii} is not affected,
      because $S$ is not changed by these rules. The fact that
      invariant~\ref{i} is maintained in Case~e) follows from
      Lemma~\ref{lem:dissub-atoms}.
      \qed
  \end{itemize}

\noindent Now we show that the nondeterministic rules can be applied in such a way that
the invariants are maintained and the application does not lead to failure.

\begin{lem}
  If \sub is an unsolved (dis)subsumption of $\Gamma$ to which no eager rule
  applies, then there is a nondeterministic rule that can be successfully
  applied to~\sub while maintaining the invariants.
\end{lem}
\proof
  If \sub is an unsolved subsumption, then it is of the form
  $C_1 \sqcap \dots \sqcap  C_n \sqsubseteq^? D$, where $D$ is a non-variable
  atom.
  By invariant~\ref{i}, we have
  $\sigma(C_1) \sqcap \dots \sqcap \sigma(C_n) \sqsubseteq \sigma(D)$.
  By Lemma~\ref{lem:sub}, there is an index $i\in\{1,\dots,n\}$ and a top-level
  atom~$E$ of $\sigma(C_i)$ such that $E\sqsubseteq\sigma(D)$.
  \begin{itemize}
    \item If $C_i$ is a constant, then by Lemma~\ref{lem:sub} we have
      $C_i=E=D$, and thus \rulefont{Eager Solving} is applicable, which
      contradicts the assumption.
    \item If $C_i=\exists r.C'$, then $\sigma(C_i)=\exists r.\sigma(C')=E$ and
      by Lemma~\ref{lem:sub} we must have $D=\exists r.D'$ and
      $\sigma(C')\sqsubseteq\sigma(D')$. Thus, the \rulefont{Decomposition}
      rule can be successfully applied to~\sub and results in a new
      subsumption $C'\sqsubseteq^?D'$ that is solved by~$\sigma$.
    \item If $C_i$ is a variable, then invariant~\ref{ii} is preserved by
      adding $D$ to $S_{C_i}$ since we have that
      $\sigma(C_i)\sqsubseteq E\sqsubseteq\sigma(D)$. Thus, we can successfully
      apply the \rulefont{Extension} rule to~\sub.
  \end{itemize}\medskip

\noindent  If \sub is an unsolved dissubsumption, then it must be of the form
  $C_1 \sqcap \dots \sqcap  C_n \nsqsubseteq^? X$ since otherwise one of the
  eager rules in Figure~\ref{fig:eager-rules} would be applicable to it. We
  have $\sigma(C_1) \sqcap \dots \sqcap  \sigma(C_n) \nsqsubseteq \sigma(X)$ by
  invariant~\ref{i}.
  By Lemma~\ref{lem:dissub}, there is a top-level atom~$E$ of~$\sigma(X)$ such
  that $\sigma(C_1) \sqcap \dots \sqcap  \sigma(C_n) \nsqsubseteq E$.
  Since $\sigma$ is local, we must have $E = \sigma(D)$ for some $D \in \NV$.
  Hence, adding $D$ to~$S_X$ maintains invariant~\ref{ii}, and adding
  $C_1 \sqcap \dots \sqcap  C_n \nsqsubseteq^? D$ to $\Gamma$ maintains
  invariant~\ref{i}.
  Thus, we can successfully apply the \rulefont{Local Extension} rule to~\sub.
\qed

This concludes the proof of correctness of
Algorithm~\ref{alg:rules-local-disunification}.
Moreover, together with Lemma~\ref{lem:dismatching-termination}, we obtain an
alternative proof of Fact~\ref{fact:np}.

\begin{thm}
\label{thm:rules-correct}
  The flat disunification problem~$\Gamma_0$ has a local solution iff there is
  a successful run of Algorithm~\ref{alg:rules-local-disunification}
  on~$\Gamma_0$.
\qed
\end{thm}
We have restricted the nondeterministic choices of
Algorithm~\ref{alg:rules-local-disunification} in such way that non-variable
atoms are only added to the assignment~$S$ if this is necessary to directly
solve some (dis)subsumption in~$\Gamma$.
Hence, the algorithm cannot be used to compute \emph{all} local solutions
of~$\Gamma$, but already selects the more ``interesting'' ones.
As described in the introduction, further dissubsumptions of the form
$X\nsqsubseteq^?D$ with $X\in\Var$ and $D\in\NV$ can be added to~$\Gamma$ in
order to further restrict the solution space.

\section{Encoding local disunification into SAT}
\label{sec:sat}

In the following, we consider an alternative algorithm for local disunification
that is based on a polynomial encoding into a SAT problem.
This reduction is a generalization of the one developed for unification
problems in~\cite{BaMo-LPAR10}.
We again consider a flat disunification problem~$\Gamma$ and the sets \At,
\Var, and \NV as in Section~\ref{sec:disunification}.
Since we are restricting our considerations to \emph{local} solutions, we can
without loss of generality assume that the sets \Nv, \Nc, and \NR contain
exactly the variables, constants, and role names occurring in~$\Gamma$.
To further simplify the reduction, we assume in the following that all flat
dissubsumptions in~$\Gamma$ are of the form $X\nsqsubseteq^?Y$ for variables
$X,Y$.
This is without loss of generality, which can be shown using a transformation
similar to that of Section~\ref{sec:flat-disunification}.

The translation uses the propositional variables $[C\sqsubseteq D]$ for all
$C,D\in\At$.
The SAT problem consists of a set of clauses $\Cl(\Gamma)$ over these variables
that express properties of (dis)subsumption in \EL and encode the elements
of $\Gamma$.
The intuition is that a satisfying valuation of $\Cl(\Gamma)$ induces a local
solution $\sigma$ of $\Gamma$ such that $\sigma(C)\sqsubseteq\sigma(D)$ holds
whenever $[C\sqsubseteq D]$ is true under the valuation.
The solution $\sigma$ is constructed by first extracting an acyclic assignment
$S$ out of the satisfying valuation and then computing $\sigma:=\sigma_S$.
We additionally introduce the variables $[X>Y]$ for all $X,Y\in\Nv$ to ensure
that the generated assignment $S$ is indeed acyclic. This is achieved by
adding clauses to $\Cl(\Gamma)$ that express that $>_S$ is a strict partial
order, \ie irreflexive and transitive.

We further use the auxiliary variables $p_{C,X,D}$ for all $X\in\Nv$,
$C\in\At$, and $D\in\NV$ to express the restrictions imposed by dissubsumptions
of the form $C\nsqsubseteq^?X$ in clausal form.
More precisely, whenever $[C\sqsubseteq X]$ is false for some $X\in\Nv$ and
$C\in\At$, then the dissubsumption $\sigma(C)\nsqsubseteq\sigma(X)$ should 
hold.
By Lemma~\ref{lem:dissub}, this means that we need to find an atom $D\in\NV$
that is a top-level atom of $\sigma(X)$ and satisfies
$\sigma(C)\nsqsubseteq\sigma(D)$.
This is enforced by making the auxiliary variable $p_{C,X,D}$ true, which makes
$[X\sqsubseteq D]$ true and $[C\sqsubseteq D]$ false (see
Definition~\ref{def:sat}\ref{cl:dissub} and Lemma~\ref{lem:soundness:dissub}
for details).

To denote propositional clauses, we use the implicative form $\phi\to\psi$,
where $\phi$ is the conjunction of all negative literals of the clause, and
$\psi$ is the disjunction of all positive literals.
We use $\top$ to denote an empty conjunction, and $\bot$ for an empty
disjunction.

\begin{defi}
\label{def:sat}
  The set $\Cl(\Gamma)$ contains the following propositional clauses:
  \begin{enumerate}[label=(\Roman*),widest=IV]
    \item\label{cl:gamma}
      \emph{Translation of $\Gamma$.}
      \begin{enumerate}[label=\alph*.,ref=\theenumi\alph*]
        \item\label{cl:gp:nonvar} For every subsumption
          $C_1\sqcap\dots\sqcap C_n\sqsubseteq^?D$ in $\Gamma$ with $D\in\NV$:
          \[ \top\to[C_1\sqsubseteq D]\lor\dots\lor[C_n\sqsubseteq D] \]
        \item\label{cl:gp:var} For every subsumption
          $C_1\sqcap\dots\sqcap C_n\sqsubseteq^?X$ in $\Gamma$ with $X\in\Nv$,
          and every $E\in\NV$:
          \[ [X\sqsubseteq E]\to
            [C_1\sqsubseteq E]\lor\dots\lor[C_n\sqsubseteq E] \]
        \item\label{cl:gm} For every dissubsumption $X\nsqsubseteq^?Y$ in
          $\Gamma$: \  $[X\sqsubseteq Y]\to\bot$
      \end{enumerate}
    \item\label{cl:sub}
      \emph{Properties of subsumptions between non-variable atoms.}
      \begin{enumerate}[label=\alph*.,ref=\theenumi\alph*]
        \item\label{cl:sub:c1} For every $A\in\Nc$: \ $\top\to[A\sqsubseteq A]$
        \item\label{cl:sub:c2} For every $A,B\in\Nc$ with $A\neq B$: \
          $[A\sqsubseteq B]\to\bot$
        \item\label{cl:sub:e1} For every $\exists r.A,\exists s.B\in\NV$ with
          $r\neq s$: \ $[\exists r.A\sqsubseteq \exists s.B]\to\bot$
        \item\label{cl:sub:e2} For every $A\in\Nc$ and $\exists r.B\in\NV$:
          \[ [A\sqsubseteq\exists r.B]\to\bot
            \quad\text{and}\quad
            [\exists r.B\sqsubseteq A]\to\bot \]
        \item\label{cl:sub:decomp}
          For every $\exists r.A,\exists r.B\in\NV$:
          \[ [\exists r.A\sqsubseteq\exists r.B]\to[A\sqsubseteq B]
            \quad\text{and}\quad
            [A\sqsubseteq B]\to[\exists r.A\sqsubseteq\exists r.B] \]
      \end{enumerate}
    \item\label{cl:sub:trans}
      \emph{Transitivity of subsumption.}\\
          For every $C_1,C_2,C_3\in\At$: \
          $[C_1\sqsubseteq C_2]\land[C_2\sqsubseteq C_3]\to
            [C_1\sqsubseteq C_3]$
    \item\label{cl:dissub}
      \emph{Dissubsumptions of the form $C\nsqsubseteq^?X$ with a
        variable~$X$.}\\
        For every $C\in\At$, $X\in\Nv$:
        \[ \top\to[C\sqsubseteq X]\lor\bigvee_{D\in\NV}p_{C,X,D}, \]
        and additionally for every $D\in\NV$:
        \[ p_{C,X,D}\to[X\sqsubseteq D]
          \quad\text{and}\quad
          p_{C,X,D}\land[C\sqsubseteq D]\to\bot \]
    \item\label{cl:order}
      \emph{Properties of $>$.}
      \begin{enumerate}[label=\alph*.,ref=\theenumi\alph*]
        \item\label{cl:irrefl}
          For every $X\in\Nv$: \  $[X>X]\to\bot$
        \item\label{cl:trans}
          For every $X,Y,Z\in\Nv$: \  $[X>Y]\land[Y>Z]\to[X>Z]$
        \item\label{cl:conn}
          For every $X,Y\in\Nv$ and $\exists r.Y\in\At$: \
          $[X\sqsubseteq\exists r.Y]\to[X>Y]$
      \end{enumerate}
  \end{enumerate}
\end{defi}\medskip

\noindent The main difference to the encoding in~\cite{BaMo-LPAR10} lies in the
clauses~\ref{cl:dissub} that ensure the presence of a non-variable atom~$D$
that solves the dissubsumption $C\nsqsubseteq^?X$ (cf.\ 
Lemma~\ref{lem:dissub}).
We also need some additional clauses in~\ref{cl:sub} to deal with
dissubsumptions.
It is easy to see that $\Cl(\Gamma)$ can be constructed in time cubic in the 
size
of~$\Gamma$ (due to the clauses in~\ref{cl:sub:trans} and~\ref{cl:trans}). We
prove the correctness of this reduction in the following two sections.

\subsection{Soundness}
\label{sec:sat-soundness}

Let $\tau$ be a valuation of the propositional variables that satisfies
$\Cl(\Gamma)$. We define the assignment $S^\tau$ as follows:
\[ S^\tau_X:=\{D\in\NV \mid \tau([X\sqsubseteq D])=1 \}. \]
We show the following property of~$>_{S^\tau}$; the proof is exactly the same
as in~\cite{BaMo-LPAR10}, but uses a different notation.

\begin{lem}
\label{lem:acyclic}
  The relation $>_{S^\tau}$ is irreflexive.
\end{lem}
\proof
  We first show that $X>_{S^\tau}Y$ implies $\tau([X>Y])=1$ for all
  $X,Y\in\Nv$.
  If $Y$ occurs in an atom of $S^\tau_X$, then this atom must be of the form
  $\exists r.Y$ with $r\in\NR$. By construction of $S^\tau$, this implies that
  $\tau([X\sqsubseteq\exists r.Y])=1$. Since $\tau$ satisfies the clauses in
  \ref{cl:conn}, we have $\tau([X>Y])=1$.
  By definition of $>_{S^\tau}$ and the transitivity clauses in
  \ref{cl:trans}, we conclude that $\tau([X>Y])=1$ whenever $X>_{S^\tau}Y$.

  Assume now that $X>_{S^\tau}X$ holds for some $X\in\Nv$. By the claim above,
  this implies that  $\tau([X>X])=1$. But this is impossible since $\tau$
  satisfies the clauses in~\ref{cl:irrefl}.
\qed

This in particular shows that $S^\tau$ is acyclic. In the following, let
$\sigma_\tau$ denote the substitution $\sigma_{S^\tau}$ induced by $S^\tau$.
We show that $\sigma_\tau$ is a solution of $\Gamma$.

\begin{lem}
\label{lem:soundness:sub}
  If $C,D\in\At$ such that $\tau([C\sqsubseteq D])=1$, then
  $\sigma_\tau(C)\sqsubseteq\sigma_\tau(D)$.
\end{lem}
\proof
  We show this by induction on the pairs
  $(\rd(\sigma_\tau(D)),\Var(D))$, where $\Var(D)$ is either the variable that
  occurs in $D$, or $\bot$ if $D$ is ground. These pairs are compared by the
  lexicographic extension of the order $>$ on natural numbers for the first
  component and the order $>_{S^\tau}$ for the second component, which is
  extended by $Y>_{S^\tau}\bot$ for all $Y\in\Nv$.

  We make a case distinction on the form of $C$ and $D$ and consider first the
  case that $D$ is a variable. Let $\sigma_\tau(E)$ be any top-level atom of
  $\sigma_\tau(D)$, which means that $\tau([D\sqsubseteq E])=1$. By the clauses
  in~\ref{cl:sub:trans}, we also have $\tau([C\sqsubseteq E])=1$. Since
  $\rd(\sigma_\tau(D))\geq \rd(\sigma_\tau(E))$ and
  $\Var(D)=D>_{S^\tau}\Var(E)$, by induction we get
  $\sigma_\tau(C)\sqsubseteq\sigma_\tau(E)$. Since $\sigma_\tau(D)$ is
  equivalent to the conjunction of all its top-level atoms, by
  Lemma~\ref{lem:sub} we obtain $\sigma_\tau(C)\sqsubseteq\sigma_\tau(D)$.

  If $D$ is a non-variable atom and $C$ is a variable, then
  $\sigma_\tau(C)\sqsubseteq\sigma_\tau(D)$ holds by construction
  of~$S^\tau$ and Lemma~\ref{lem:sub}.

  If $C,D$ are both non-variable atoms, then by the clauses in~\ref{cl:sub}
  they must either be the same concept constant, or be existential restrictions
  using the same role name. In the first case, the claim follows immediately.
  In the latter case, let $C=\exists r.C'$ and $D=\exists r.D'$. By the clauses
  in~\ref{cl:sub:decomp}, we have $\tau([C'\sqsubseteq D'])=1$. Since
  $\rd(\sigma_\tau(D))>\rd(\sigma_\tau(D'))$, by induction we get
  $\sigma_\tau(C')\sqsubseteq\sigma_\tau(D')$, and thus
  $\sigma_\tau(C)\sqsubseteq\sigma_\tau(D)$ by Lemma~\ref{lem:sub}.
\qed

We now show that the converse of this lemma also holds.

\begin{lem}
\label{lem:soundness:dissub}
  If $C,D\in\At$ such that $\tau([C\sqsubseteq D])=0$, then
  $\sigma_\tau(C)\nsqsubseteq\sigma_\tau(D)$.
\end{lem}
\proof
  We show this by induction on the tuples
  $(\rd(\sigma_\tau(C)),\Var(C),\Var(D))$ and make a case distinction on the
  form of $C$ and $D$.
  If $D$ is a variable, then by the clauses in~\ref{cl:dissub} there must be a
  $D'\in\NV$ such that $\tau(p_{C,D,D'})=1$. This implies that
  $\tau([D\sqsubseteq D'])=1$ and $\tau([C\sqsubseteq D'])=0$.
  By construction of $S^\tau$, $\sigma_\tau(D')$ is a top-level atom of
  $\sigma_\tau(D)$ and $\Var(D)>_{S^\tau}\Var(D')$. Since
  $\rd(\sigma_\tau(C))=\rd(\sigma_\tau(C))$ and $\Var(C)=\Var(C)$, by induction
  we get $\sigma_\tau(C)\nsqsubseteq\sigma_\tau(D')$,
  and thus $\sigma_\tau(C)\nsqsubseteq\sigma_\tau(D)$ by
  Lemma~\ref{lem:dissub}.

  If $D$ is a non-variable atom and $C$ is a variable, then consider any
  top-level atom $\sigma_\tau(E)$ of $\sigma_\tau(C)$, which means that we have
  $\tau([C\sqsubseteq E])=1$. By the clauses in~\ref{cl:sub:trans} this
  implies that $\tau([E\sqsubseteq D])=0$. Since we have
  $\rd(\sigma_\tau(C))\geq\rd(\sigma_\tau(E))$ and
  $\Var(C)=C>_{S^\tau}\Var(E)$, by induction we get
  $\sigma_\tau(E)\nsqsubseteq\sigma_\tau(D)$. Since $\sigma_\tau(C)$ is
  equivalent to the conjunction of all its top-level atoms, by
  Lemma~\ref{lem:dissub} we get $\sigma_\tau(C)\nsqsubseteq\sigma_\tau(D)$.

  If $C,D$ are both non-variable atoms, then by the clauses in~\ref{cl:sub},
  they are either different constants, a constant and an existential
  restriction, or two existential restrictions. In the first two cases,
  $\sigma_\tau(C)\nsqsubseteq\sigma_\tau(D)$ holds by Lemma~\ref{lem:sub}. In
  the last case, they can either contain two different roles or the same role.
  Again, the former case is covered by Lemma~\ref{lem:sub}, while in the latter
  case we have $C=\exists r.C'$, $D=\exists r.D'$, and
  $\tau([C'\sqsubseteq D'])=0$ by the clauses in~\ref{cl:sub:decomp}. Since
  $\rd(\sigma_\tau(C))>\rd(\sigma_\tau(C'))$, by induction we get
  $\sigma_\tau(C')\nsqsubseteq\sigma_\tau(D')$, and thus
  $\sigma_\tau(C)\nsqsubseteq\sigma_\tau(D)$ by Lemma~\ref{lem:dissub}.
\qed

This suffices to show soundness of the reduction.

\begin{lem}
\label{lem:soundness}
  The local substitution~$\sigma_\tau$ solves~$\Gamma$.
\end{lem}
\proof
  Consider any flat subsumption $C_1\sqcap\dots\sqcap C_n\sqsubseteq^?D$
  in~$\Gamma$. If $D\in\NV$, then we have
  $\sigma_\tau(C_i)\sqsubseteq\sigma_\tau(D)$ for some $i$, $1\le i\le n$, by
  the clauses in~\ref{cl:gamma} and Lemma~\ref{lem:soundness:sub}. By
  Lemma~\ref{lem:sub}, $\sigma_\tau$ solves the subsumption.

  If $D$ is a variable, then consider any top-level atom $\sigma_\tau(E)$ of
  $\sigma_\tau(D)$, for which we must have $\tau([D\sqsubseteq E])=1$. By the
  clauses in~\ref{cl:gamma}, there must be an $i$, $1\leq i\leq n,$ such that
  $\tau([C_i\sqsubseteq E])=1$, and thus
  $\sigma_\tau(C_i)\sqsubseteq\sigma_\tau(E)$ by
  Lemma~\ref{lem:soundness:sub}. Again, by Lemma~\ref{lem:sub} this implies
  that $\sigma_\tau$ solves the subsumption.

  Finally, consider a dissubsumption $X\nsqsubseteq^?Y$ in $\Gamma$. Then by
  the clauses in~\ref{cl:gamma} and Lemma~\ref{lem:soundness:dissub} we have
  $\sigma_\tau(X)\nsqsubseteq\sigma_\tau(Y)$ \ie $\sigma_\tau$ solves the
  dissubsumption.
\qed

\subsection{Completeness}
\label{sec:sat-completeness}

Let now $\sigma$ be a ground local solution of $\Gamma$ and $>_\sigma$ the
resulting partial order on~\Nv, defined as follows for all $X,Y\in\Nv$:
\[ X>_\sigma Y \text{ iff } \sigma(X)\sqsubseteq\exists r_1.\dots\exists
r_n.\sigma(Y) \text{ for some } r_1,\dots,r_n\in\NR \text{ with } n\ge 1. \]
Note that $>_\sigma$ is irreflexive since $X>_\sigma X$ is impossible by
Lemma~\ref{lem:sub}, and it is transitive since $\sqsubseteq$ is transitive
and closed under applying existential restrictions on both sides. Thus,
$>_\sigma$ is a strict partial order.
We define a valuation~$\tau_\sigma$ as follows for all $C,D\in\At$, $E\in\NV$,
and $X,Y\in\Nv$:
\begin{align*}
  \tau_\sigma([C\sqsubseteq D]) &:= \begin{cases}
    1 &\text{if $\sigma(C)\sqsubseteq\sigma(D)$}\\
    0 &\text{otherwise}
  \end{cases} \\
  \tau_\sigma(p_{C,X,E}) &:= \begin{cases}
    1 &\text{if $\sigma(X)\sqsubseteq\sigma(E)$ and
      $\sigma(C)\nsqsubseteq\sigma(E)$}\\
    0 &\text{otherwise}
  \end{cases}\\
  \tau_\sigma([X>Y]) &:= \begin{cases}
    1 &\text{if $X>_\sigma Y$}\\
    0 &\text{otherwise}
  \end{cases}
\end{align*}

\begin{lem}
  The valuation $\tau_\sigma$ satisfies all clauses of~$\Cl(\Gamma)$.
\end{lem}
\proof
  For~\ref{cl:gp:nonvar}, consider any flat subsumption
  $C_1\sqcap\dots\sqcap C_n\sqsubseteq^?D$ in~$\Gamma$ with $D\in\NV$. Since
  $\sigma$ solves~$\Gamma$, we have
  $\sigma(C_1)\sqcap\dots\sqcap\sigma(C_n)\sqsubseteq\sigma(D)$. Since
  $\sigma(D)$ is an atom, by Lemma~\ref{lem:sub} there must be an $i$,
  $1\le i\le n$, and a top-level atom~$E$ of~$\sigma(C_i)$ such that
  $\sigma(C_i)\sqsubseteq E\sqsubseteq\sigma(D)$. By the definition
  of~$\tau_\sigma$, this shows that $\tau_\sigma([C_i\sqsubseteq D])=1$, and
  thus the clause is satisfied.

  Consider now an arbitrary flat subsumption
  $C_1\sqcap\dots\sqcap C_n\sqsubseteq^?X$ from~$\Gamma$ where $X$ is a
  variable, and any $E\in\NV$ such that $\tau_\sigma([X\sqsubseteq E])=1$.
  This implies that we have
  $\sigma(C_1)\sqcap\dots\sqcap\sigma(C_n)
    \sqsubseteq\sigma(X)
    \sqsubseteq\sigma(E)$,
  and thus as above there is a top-level atom~$F$ of some~$\sigma(C_i)$
  such that $\sigma(C_i)\sqsubseteq F\sqsubseteq\sigma(E)$, which shows
  that $\tau_\sigma([C_i\sqsubseteq E])=1$, as required for the clause
  in~\ref{cl:gp:var}.

  For every dissubsumption $X\nsqsubseteq^?Y$ in~$\Gamma$, we must have
  $\sigma(X)\nsqsubseteq\sigma(Y)$, and thus $\tau_\sigma([X\sqsubseteq Y])=0$,
  satisfying the clause in~\ref{cl:gm}.

  For $A\in\Nc$, we have $\sigma(A)\sqsubseteq\sigma(A)$, and thus
  $\tau_\sigma([A\sqsubseteq A])=1$. Similar arguments show that the remaining
  clauses in~\ref{cl:sub} are also satisfied (see Lemma~\ref{lem:sub}).
  For~\ref{cl:sub:trans}, consider $C_1,C_2,C_3\in\At$ with
  $\tau_\sigma([C_1\sqsubseteq C_2])=\tau_\sigma([C_2\sqsubseteq C_3])=1$, and
  thus $\sigma(C_1)\sqsubseteq\sigma(C_2)\sqsubseteq\sigma(C_3)$. By
  transitivity of~$\sqsubseteq$, we infer
  $\tau_\sigma([C_1\sqsubseteq C_3])=1$.

  For every $C\in\At$, $X\in\Nv$, and $D\in\NV$ with $\tau_\sigma(p_{C,X,D})=1$,
  we must have $\tau_\sigma([X\sqsubseteq D])=1$ and
  $\tau_\sigma([C\sqsubseteq D])=0$ by the definition of~$\tau_\sigma$.
  Furthermore, whenever $\tau_\sigma([C\sqsubseteq X])=0$, we have
  $\sigma(C)\nsqsubseteq\sigma(X)$, and thus by Lemma~\ref{lem:dissub} there
  must be a top-level atom~$E$ of $\sigma(X)$ such that
  $\sigma(C)\nsqsubseteq E$. Since $\sigma$ is a local solution, $E$ must be of
  the form $\sigma(F)$ for some $F\in\NV$, and thus we obtain
  $\sigma(X)\sqsubseteq\sigma(F)$ and $\sigma(C)\nsqsubseteq\sigma(F)$, and
  hence $\tau_\sigma(p_{C,X,F})=1$.
  This shows that all clauses in~\ref{cl:dissub} are satisfied
  by~$\tau_\sigma$.

  For~\ref{cl:irrefl}, recall that $>_\sigma$ is irreflexive. Transitivity
  of~$>_\sigma$ yields satisfaction of the clauses in~\ref{cl:trans}.
  Finally, if $\sigma(X)\sqsubseteq\sigma(\exists r.Y)=\exists r.\sigma(Y)$ for
  some $X,Y\in\Nv$ with $\exists r.Y\in\At$, we have $X>_\sigma Y$ by
  definition, and thus the clauses in~\ref{cl:conn} are satisfied
  by~$\tau_\sigma$.
\qed

This completes the proof of the correctness of the translation presented in
Definition~\ref{def:sat}, which provides us with a reduction of local
disunification (and thus also of dismatching) to SAT.
Since the size of $\Cl(\Gamma)$ is polynomial in the size of~$\Gamma$, we
obtain yet another proof of Fact~\ref{fact:np}.

\begin{thm}
\label{thm:sat-correct}
  The flat disunification problem~$\Gamma$ has a local solution iff
  $\Cl(\Gamma)$ is satisfiable.
\qed
\end{thm}
Regarding the computation of actual solutions, we note that
the definition of $S^\tau$ in Section~\ref{sec:sat-soundness}
describes how to obtain local solutions of~$\Gamma$ from the satisfying
valuations of~$\Cl(\Gamma)$.
From a syntactic point of view, this approach does not yield all local solutions. 
In fact, the
transitivity clauses~\ref{cl:sub:trans} may force us to add atoms
to~$S^\tau$ that are, syntactically, not necessary to solve~$\Gamma$.
Also note that different satisfying valuations~$\tau$ may sometimes yield
equivalent unifiers, because some atoms in the substitution~$\sigma_\tau(X)$
of a variable~$X$ may be subsumed by others.
Nevertheless, we can show that, by applying the construction of Section~\ref{sec:sat-soundness}
to the satisfying valuations of $\Cl(\Gamma)$,
we obtain \emph{all}
local solutions of~$\Gamma$ \emph{modulo equivalence}. We call two
solutions~$\sigma$ and $\gamma$ \emph{equivalent} if $\sigma(X)\equiv\gamma(X)$
holds for all $X\in\Nv$.

\begin{lem}
\label{lem:sat-all-solutions}
  Let $\sigma$ be a local solution of the flat disunification problem~$\Gamma$.
  Then there is a satisfying valuation~$\tau$ of $\Cl(\Gamma)$ such that
  $\sigma_{S^\tau}$ is equivalent to~$\sigma$.
\end{lem}
\proof
  Let $S$ be the acyclic assignment underlying~$\sigma$, $\tau:=\tau_\sigma$ be
  the satisfying valuation induced by~$\sigma$ as defined in
  Section~\ref{sec:sat-completeness}, and $S^\tau$ and
  $\gamma:=\sigma_{S^\tau}$ be as defined in Section~\ref{sec:sat-soundness}.
  We first show that $S_X\subseteq S^\tau_X$ holds for all $X\in\Nv$.
  To this end, consider any non-variable atom $D\in S_X$.
  Since $\sigma(D)$ is a top-level atom of $\sigma(X)$, by Lemma~\ref{lem:sub}
  we have $\sigma(X)\sqsubseteq\sigma(D)$.
  Hence, the definitions of~$\tau$ and~$S^\tau$ yield that
  $\tau([X\sqsubseteq D])=1$ and $D\in S^\tau_X$, as required.

  We can now show by induction on the well-founded strict partial
  order~$>_{S^\tau}$ that $\sigma(X)\equiv\gamma(X)$ holds for all $X\in\Nv$.
  Assume that $\sigma(Y)\equiv\gamma(Y)$ holds for all variables
  $Y <_{S^\tau} X$, and hence $\sigma(D)\equiv\gamma(D)$ holds for all
  non-variable atoms $D\in S^\tau_X$, including those in~$S_X$ (it trivially
  holds if $D$ is ground).
  Since $\sigma(X)$ consists exactly of the top-level atoms $\sigma(D)$,
  $D\in S_X$, and similarly $\gamma(X)$ consists exactly of the top-level atoms
  $\gamma(D)$, $D\in S^\tau_X$, we thus know that each top-level atom
  of~$\sigma(X)$ is equivalent to a top-level atom of~$\gamma(X)$.
  Hence, $\gamma(X)\sqsubseteq\sigma(X)$ holds by Lemma~\ref{lem:sub}.
  For the other direction, consider any top-level atom of~$\gamma(X)$, which
  must be of the form~$\gamma(D)$ with $D\in S^\tau_X$.
  By the definition of~$S^\tau$, we obtain $\tau([X\sqsubseteq D])=1$, which
  yields $\sigma(X)\sqsubseteq\sigma(D)$ by the definition of~$\tau$.
  Hence, there must be a top-level atom of $\sigma(X)$ that is subsumed by
  $\gamma(D)\equiv\sigma(D)$, and thus Lemma~\ref{lem:sub} yields
  $\sigma(X)\sqsubseteq\gamma(X)$.
\qed

The SAT reduction has been implemented in our prototype system~UEL,%
\footnote{version 1.4.0, available at \url{http://uel.sourceforge.net/}}
which uses SAT4J%
\footnote{\url{http://www.sat4j.org/}}
as external SAT solver.
First experiments show that disunification is indeed helpful for reducing the
number and the size of solutions.
For example, a slightly modified version of the example from the introduction
has $128$ solutions without any dissubsumptions (see~\cite{BBMM-DL12} for more
details).
Each additional dissubsumption disallowing a particular non-variable atom in
the assignments (\eg the dissubsumption~\eqref{dissubs:ex} from the
introduction) roughly halves the number of remaining solutions.
The runtime performance of the solver for local disunification problems is
comparable to the one for pure unification problems, even on larger problems.

\section{Related work}
\label{sec:related-work}

Since Description Logics and Modal Logics are closely
related~\cite{Schi-IJCAI91}, results on unification in one of these two areas
carry over to the other one.
In Modal Logics, unification has mostly been considered for expressive logics
with all Boolean operators~\cite{Ghil97,Ghil99,Ryba08}. An important open
problem in the area is the question whether unification in the basic modal
logic~{\sf K}, which corresponds to the DL~$\mathcal{ALC}$, is decidable. It is
only known that relatively minor extensions of~{\sf K} have an undecidable
unification problem~\cite{WoZa-TOCL08}.

Disunification also plays an important role in Modal Logics since it is
basically the same as the admissibility problem for inference rules
\cite{Ryba97,IeMe09,BRST09}.
To be more precise, a normal modal logic~$L$ induces an equational
theory~$E_L$ that axiomatizes equivalence in this logic, where the formulas are
viewed as terms. Validity is then just equivalence to~$\top$ and inconsistency
is equivalence to~$\bot$. An \emph{inference rule} is of the form
\begin{equation}
\label{inference-rule}
  \frac{A_1,\ldots,A_m}{B_1,\ldots,B_n}
\end{equation}
where $A_1,\ldots,B_n$ are formulas (terms) that may contain variables. More 
precisely, it is not a single rule but a rule schema that stands for all its
instances
\begin{equation}
\label{inference-rule-instances}
  \frac{\sigma(A_1),\ldots,\sigma(A_m)}{\sigma(B_1),\ldots,\sigma(B_n)}
\end{equation}
where $\sigma$ is a substitution. The semantics of such a
rule~\eqref{inference-rule-instances} is the following: whenever all of its
premises are valid, then one of the consequences must be valid as well. We only
admit the inference rule~\eqref{inference-rule} for the logic~$L$ if all its
instances~\eqref{inference-rule-instances} satisfy this requirement. Thus, we
say that the inference rule~\eqref{inference-rule} is \emph{admissible for~$L$}
if
\[
  \sigma(A_1) =_{E_L} \top \wedge\ldots\wedge \sigma(A_m) =_{E_L} \top
  \ \ \mbox{implies}\ \
  \sigma(B_1) =_{E_L} \top \vee\ldots\vee \sigma(B_n) =_{E_L} \top
\]
for all substitutions~$\sigma$. Obviously, this is the case iff the
disunification problem
\[
  \{A_1 \equiv^? \top,\ldots,A_m \equiv^? \top,
    B_1 \nequiv^? \top,\ldots,B_n \nequiv^? \top\}
\]
does \emph{not} have a solution.

As already mentioned in the introduction,
(dis)unification in \EL is actually a
special case of (dis)unification modulo equational
theories~\cite{BuerckertBuntine94,Comon91,Como-JLC89}.
As shown in \cite{BaMo-LMCS10}, equivalence in \EL can be axiomatized
by the equational theory of semilattices
with monotone functions, which extends the theory ACUI of an
associative-commutative-idempotent
binary function symbol $\ast$ (corresponding to~$\sqcap$) with unit 
(corresponding to~$\top$) 
by unary function symbols $h_r$ (corresponding to $\exists r$)
that are monotone in the sense that they satisfy the identities $h_r(x)\ast h_r(x\ast y) = h_r(x\ast y)$.
Perhaps the closest to our present work is thus the investigation of disunification
in ACUI with free function symbols (i.e., additional function symbols of arbitrary arity
that satisfy no non-trivial identities). This problem is shown to be in \NP in~\cite{BaSc-TCS95,DoPP-C04};
NP-hardness follows from NP-hardness of ACUI-unification with free function symbols \cite{KapurNarendranJAR92}.
To be more precise, the \NP upper bound is shown in \cite{BaSc-TCS95} for the theory ACI with
free function symbols, using general combination results for disunification developed in the same article.
However, it is easy to see that the approach applied in \cite{BaSc-TCS95} also works for ACUI.
In contrast, the \NP upper bound in \cite{DoPP-C04} is shown for ACUI with free function symbols by directly
designing a dedicated algorithm for disunification in this theory.

\section{Conclusions}
\label{sec:conclusions}

We have considered disunification in the description logic \EL.
While the complexity of the general problem remains open, we have identified
two restrictions under which the complexity does not increase when compared to
plain unification in~\EL, \ie remains in \NP.
We developed a nondeterministic polynomial reduction from dismatching problems
to local disunification problems, and presented two algorithms to solve the
latter. These procedures extend known algorithms for unification in~\EL without
a large negative impact on their performance.

Regarding future work, we want to investigate the decidability and complexity
of general disunification in~\EL, and consider also the case where non-ground
solutions are allowed.
In contrast to unification, these extensions make the problem harder to solve.
From a more practical point of view, we plan to implement also the
goal-oriented algorithm for local disunification, and to evaluate the
performance of both presented algorithms on real-world problems.
In addition, we will investigate whether a reduction to answer set programming
(ASP) \cite{Bara-03,GeLi-ICLP88} rather than SAT leads to a better performance.

\bibliographystyle{plainnat}
\bibliography{ref}

\end{document}